\documentclass[review]{elsarticle}[11pt]

\usepackage{amsmath}  
\usepackage{amssymb}   
\usepackage{amsthm}  
\usepackage{lineno,hyperref}
\modulolinenumbers[5]
\newtheorem{theorem}{Theorem}
\newtheorem{lemma}{Lemma}
\newtheorem{remark}{Remark}
\usepackage{hyperref}
\usepackage{amsmath}
\usepackage{graphicx}
\usepackage{amssymb}
\usepackage{graphicx}
\usepackage[dvipsnames,usenames]{color}
\usepackage{subfigure}
\journal{Journal of \LaTeX\ Templates}

\usepackage{hyperref}
\usepackage{amsmath}
\usepackage{graphicx}
\usepackage{amssymb}
\usepackage{graphicx}
\usepackage[dvipsnames,usenames]{color}
\usepackage{subfigure}
\newcommand{\PsoneDone}{P_{s}(1,D,\{1\})}
\newcommand{\PfoneDone}{\overline{P}_{s}(1,D,\{1\})}
\newcommand{\PsoneDonetwo}{P_{s}(1,D,\{1,2\})}
\newcommand{\PfoneDonetwo}{\overline{P}_{s}(1,D,\{1,2\})}

\newcommand{\PftwoDtwo}{\overline{P}_{s}(2,D,\{2\})}
\newcommand{\PstwoDonetwo}{P_{s}(2,D,\{1,2\})}

\newcommand{\PskDk}{P_{s}(k,D,\{k\})}
\newcommand{\PfkDk}{\overline{P}_{s}(k,D,\{k\})}
\newcommand{\PskDonetwo}{P_{s}(k,D,\{1,2\})}
\newcommand{\PfDonetwo}{\overline{P}_{s}(0,D,\{1,2\})}
\newcommand{\PsoneRoneone}{P_{s}(1,R_1,\{1\})}
\newcommand{\PfoneRoneone}{\overline{P}_{s}(1,R_1,\{1\})}
\newcommand{\PsoneRoneonetwo}{P_{s}(1,R_1,\{1,2\})}
\newcommand{\PfoneRoneonetwo}{\overline{P}_{s}(1,R_1,\{1,2\})}
\newcommand{\PstwoRonetwo}{P_{s}(2,R_1,\{2\})}
\newcommand{\PftwoRonetwo}{\overline{P}_{s}(2,R_1,\{2\})}
\newcommand{\PstwoRoneonetwo}{P_{s}(2,R_1,\{1,2\})}

\newcommand{\PsoneRtwoone}{P_{s}(1,R_2,\{1\})}
\newcommand{\PfoneRtwoone}{\overline{P}_{s}(1,R_2,\{1\})}
\newcommand{\PsoneRtwoonetwo}{P_{s}(1,R_2,\{1,2\})}
\newcommand{\PfoneRtwoonetwo}{\overline{P}_{s}(1,R_2,\{1,2\})}
\newcommand{\PstwoRtwotwo}{P_{s}(2,R_2,\{2\})}
\newcommand{\PftwoRtwotwo}{\overline{P}_{s}(2,R_2,\{2\})}
\newcommand{\PstwoRtwoonetwo}{P_{s}(2,R_2,\{1,2\})}
\newcommand{\PftwoRtwoonetwo}{\overline{P}_{s}(2,R_2,\{1,2\})}
\newcommand{\PskRionetwo}{P_{s}(k,R_i,\{1,2\})}
\newcommand{\PfRionetwo}{\overline{P}_{s}(0,R_i,\{1,2\})}
\newcommand{\PskRik}{P_{s}(k,R_i,\{k\})}
\newcommand{\PfRik}{\overline{P}_{s}(0,R_i,\{k\})}

\newcommand{\PfRoneonetwo}{\overline{P}_{s}(0,R_1,\{1,2\})}
\newcommand{\PfRtwoonetwo}{\overline{P}_{s}(0,R_2,\{1,2\})}
\newcommand{\PskRonek}{P_{s}(k,R_1,\{k\})}
\newcommand{\PfRonek}{\overline{P}_{s}(0,R_1,\{k\})}
\newcommand{\PskRtwok}{P_{s}(k,R_2,\{k\})}
\newcommand{\PfRtwok}{\overline{P}_{s}(0,R_2,\{k\})}

\newcommand{\PsRoneDRone}{P_{s}(R_1,D,\{R_1\})}

\newcommand{\PsRiDRi}{P_{s}(R_i,D,\{R_i\})}
\newcommand{\PfRiDRi}{\overline{P}_{s}(R_i,D,\{R_i\})}
\newcommand{\PssRoneDRone}{P_{s}^{*}(R_1,D,\{R_1\})}

\newcommand{\PssRtwoDRtwo}{P_{s}^{*}(R_2,D,\{R_2\})}

\newcommand{\PsRoneDRoneRtwo}{P_{s}(R_1,D,\{R_1,R_2\})}

\newcommand{\PsRiDRoneRtwo}{P_{s}(R_i,D,\{R_1,R_2\})}
\newcommand{\PfRiDRoneRtwo}{\overline{P}_{s}(R_i,D,\{R_1,R_2\})}
\newcommand{\PssRiDRi}{P_{s}^{*}(R_i,D,\{R_i\})}
\newcommand{\PfsRiDRi}{\overline{P}_{s}^{*}(R_i,D,\{R_i\})}

\newcommand{\PsRtwoDRtwo}{P_{s}(R_2,D,\{R_2\})}

\newcommand{\PsRtwoDRoneRtwo}{P_{s}(R_2,D,\{R_1,R_2\})}

\newcommand{\Bone}{\alpha_{2} \left[\alpha_{1}\Delta_2 + \PsRtwoDRtwo\right]}
\newcommand{\Btwo}{\alpha_{2}\Delta_1+\PsRoneDRone}
\newcommand{\Bthree}{\alpha_{1} \left[\alpha_{2}\Delta_1 + \PsRoneDRone\right]}
\newcommand{\Bfour}{\alpha_{1}\Delta_2+\PsRtwoDRtwo}

\begin{document}

\begin{frontmatter}

\title{Performance Analysis of a Cooperative Wireless Network with Adaptive Relays}
\author{Ioannis Dimitriou\fnref{myfootnote}}
\fntext[myfootnote]{Corresponding author}
\ead{idimit@math.upatras.gr}
\address{Department of Mathematics, University of Patras, 26500 Patras, Greece}
\ead[url]{http://www.math.upatras.gr/~idimit/}

\author{Nikolaos Pappas}
\ead{nikolaos.pappas@liu.se}

\address{Department of Science and Technology, Link\"{o}ping University, Campus Norrk\"{o}ping, Sweden.}

\begin{abstract}
In this work, we investigate a slotted-time relay assisted cooperative random access wireless network with multipacket (MPR) reception capabilities. MPR refers to the capability of a wireless node to successfully receive packets from more than two other modes that transmit simultaneously at the same slot. We consider a network of $N$ saturated sources that transmit packets to a common destination node with the cooperation of two infinite capacity relay nodes. The relays assist the sources by forwarding the packets that failed to reach the destination. Moreover, the relays have also packets of their own to transmit to the destination. We further assume that the relays employ a state-dependent retransmission control mechanism. In particular, a relay node accordingly adapts its transmission probability based on the status of the other relay. Such a protocol is towards self-aware networks and leads to substantial performance gains in terms of delay. We investigate the stability region and the throughput performance for the full MPR model. Moreover, for the asymmetric two-sources, two-relay case we derive the generating function of the stationary joint queue-length distribution with the aid of the theory of boundary value problems. For the symmetric case, we obtain explicit expressions for the average queueing delay in a relay node without solving a boundary value problem. Extensive numerical examples are presented and provide insights on the system performance.\end{abstract}


\end{frontmatter}

%\linenumbers

\section{Introduction}
Over the past few decades, wireless communications and networking have witnessed an unprecedented growth. The growing demands require high data rates, considerably large coverage areas, and high reliability. Relay-assisted wireless networks have been proposed as a candidate solution ot fulfill these requirements \cite{rab}, since relays can decrease the delay and can also provide increased reliability and higher energy efficiency \cite{cov,nos}. A relay-based cooperative wireless system operates as follows: there is a finite number of sources that transmit packets to a common destination node, and a finite number of relay nodes that assist the sources by storing and retransmitting the packets that failed to reach the destination; e.g., \cite{sad, PappasGC2010,PappasGlobalSIP2013, PappasTWC2015,pap}. 

A cooperation strategy among sources and relays specifies which of the relays will cooperate with the sources. This problem gives rise to the usage of a cooperative space diversity protocol \cite{send}, where each source has a number of ``partners" (i.e., relays) that are responsible for retransmitting its failed packets. Recent advances in IoT networks with multiple nodes and, multiple relays, reveals the challenging task to choose the subset of relays to cooperate in order optimize the network performance, see e.g., \cite{palombara,LiangLu,Ozfatura2018,yang2009,zanella2017}.

\subsection{Related work}
Cooperative relaying is mostly considered at the physical layer, and is based on information-theoretic considerations. The classical relay channel was first examined in \cite{mul} and later in \cite{cov}. Recently, cooperative communications have received renewed attention, as a powerful technique to combat fading and attenuation in wireless networks; e.g., \cite{send,lan}. Most of the research has concentrated on information-theoretic studies. Recent works \cite{PappasTWC2015,pap,sad,rong} shown that similar gains can be achieved by network-layer cooperation. By network-layer cooperation, relaying is assumed to take place at the protocol level avoiding physical layer considerations.

In addition, random access recently re-gained interest due to the increased number of communicating devices in 5G networks, and the need for massive uncoordinated access \cite{Laya2014}. Random access and alternatives schemes and their effect on the operation of LTE and LTE-A are presented in \cite{Laya2014}, \cite{KoseogluTCOM2016}, \cite{LeungTWC2012}. In \cite{PopovskiSPL2017}, the effect of random access in Cloud-Radio Access Network is considered. A Markov chain model for
the calculation of the average cooperation delay in persistent relay carrier sensing multiple access scheme is presented in \cite{alonso1}, while in \cite{alonso2} three on-demand cooperation strategies were proposed to manage the multiple relay
access control problem to handle relay transmissions in an effective manner; see also \cite{alonso3}. An analytical cross-layer framework to model end-to-end metrics (i.e., throughput and
energy efficiency), in two-way cooperative networks subject to correlated shadowing conditions is proposed in \cite{antonopoulos}.

The vast majority of the related literature focused on the characterization of the stable throughput region, i.e. the stability region, which gives the set of arrival rates such that there exist transmission probabilities under which the system is stable. Clearly, this is a meaningful metric to measure the impact of bursty traffic and the concept of interacting nodes in a network; e.g., \cite{LuoAE1999,Rao_TIT1988,szpa}. In addition to throughput, delay is another important metric that recently received considerable attention due to the development of 5G and beyond networks and the rapid growth on supporting real-time applications, which in turn require delay-based guarantees; e.g., \cite{Laya2014,KoseogluTCOM2016}. 

However, due to the interdependence among queues, the characterization of the delay even in small networks with random access is a rather difficult task, even for the non-cooperative collision channel model \cite{nain}. 
In the collision channel model when more than one nodes transmit simultaneously, none of them will be successful. Such a channel model is accurate for modeling wire-line communication, but it is not appropriate to model the probabilistic reception in wireless multiple access networks. MPR access scheme have been recently developed to cope with the probabilistic reception in wireless systems. Under such a scheme, we can have successful transmissions even if more than one nodes transmit simultaneously. For the \textbf{non}-cooperative multipacket reception model, delay analysis was performed in \cite{NawareTong2005}, based on the assumption of a symmetric network. Recently, the authors in \cite{dimpap} generalized the model in \cite{NawareTong2005,nain} by considering time-varying links between nodes where the channel state information was modeled according to a Gilbert-Elliot model, and obtain explicit expressions for the queueing delay in terms of the solution of a boundary value problem; see also \cite{dimpaptwc,PappasITC30}. 

The study of queueing systems using the theory of boundary value problems was initiated in \cite{fay1}, and a concrete methodological approach was given in \cite{coh,fay}. The vast majority of queueing models are analyzed with the aid of the theory of boundary value problems referring to continuous time models, e.g., \cite{avr,box,dim1,dim2,fr,guillemin2004,Guillemin2013,van}. On the contrary, there are very few works on the analysis of discrete time models \cite{nain,dimpaptwc,szpa,szpa1,szpa2}. This is mainly due to the complex behavior of the underlying random walk, which reflects the interdependence and coupling among the queues.  
\subsection{Contribution}
Our contribution is summarized as follows. We consider a cooperative wireless network with $N$ saturated sources, two relay nodes with adaptive transmission control, and a common destination. Our primary interest is to investigate the stability conditions, the throughput performance, and the queueing delay experienced at the buffers of relay nodes. The time is slotted, corresponding to the duration of a transmission of a packet, and the sources/relay nodes access the medium in a random access manner. The sources transmit packets to the destination with
the cooperation of the two relays. If a direct transmission of a source's packet to the destination fails, the relays store it in their queues acoording to a probabilistic policy, and try to forward it to the destination at a subsequent time slot. Moreover, the relays have also external bursty arrivals that are stored in their infinite capacity queues. We consider MPR capabilities at the destination node.

We assume that sources and relays transmit in different channels, but the destination node overhears both of them. In particular, the destination node first senses the sources. If it senses no activity from the sources, it switches to the relays. For modeling reasons and to enhance the readability, we assume that the ``sensing" time is negligible. The relays are accessing the wireless channel randomly and employ a state-dependent transmission protocol. More precisely, a relay adapts its transmission characteristics based on the status of the other relay in order to exploit its idle slots, and to increase its transmission efficiency, which in turn leads towards self-aware networks \cite{mah}. Note that this feature is common in cognitive radios \cite{sad,mah}. To the best of our knowledge this variation of random access has not been reported in the literature. Our contribution is twofold. The first part focused on the stable throughput region, and the second on the queueing delay analysis at relay nodes.

\subsubsection{Stability analysis and throughput performance}
We provide the throughput analysis of the asymmetric two-sources random access network, and the symmetric $N$-sources random access network, both under MPR channel model. The performance characterization for $N$ symmetric sources can provide insights on scalability of the network. In addition, we provide the stability conditions for the queues at the relays. 
 
\subsubsection{Delay Analysis}
The second part of our contribution refers to the delay analysis. Except its practical implications, our work is also theoretically oriented. To the best of our knowledge, there is no other work in the related literature that deals with the detailed delay analysis of an asymmetric random access cooperative wireless system with adaptive transmissions and MPR capabilities. In particular, we consider a Markov chain in the positive quadrant which posses a partial spatial homogeneity, which in turn, reflects the ability of the relays to adapt their re-transmission capabilities.   

To enhance the readability of our work we consider the case of $N=2$ sources, and focus on a subclass of MPR models, called the ``capture" channel, under which at most one packet can be successfully decoded by the receiver of the node $D$, even if more than one nodes transmit\footnote{Recent advances in IoT, LoRaWAN \cite{bankov} reveals the importance and applicability of the capture  channel model.}. 

We have to mention, that the assumption of two sources is not restrictive, and our analysis can be extended to the general case of $N$ sources. Moreover, our analysis remains valid even for the case of general MPR model. However, in both cases some important technical requirements must be further taken into account, which in turn will worse the readability of the paper and further increase its length. Besides, our aim here is to focus on the fundamental problem of characterizing the delay in a cooperative wireless network with two relays, and our model and its analysis serve as a building block for more general cases. 

Our system is modeled as a two-dimensional discrete time Markov chain, and we derive the generating function of the stationary joint relay queue length distribution in terms of the solution of a Riemann-Hilbert boundary value problem. Furthermore, each relay node employs a state-dependent transmission policy, under which it adapts its transmission probabilities based on the status of the queue of the other relay. In addition, the kernel of this functional equation has never been treated in the related literature. More precisely,
\begin{itemize}
\item Based on a relation among the values of the transmission probabilities we distinguish the analysis in two cases, which are different both in the modeling, and in the technical point of view. In particular, the analysis leads to the formulation of two boundary value problems \cite{ga} (i.e., a Dirichlet, and a Riemann-Hilbert problem), the solution of which provides the generating function of the stationary joint distribution of the queue size for the relays. This is the key element for obtaining expressions for the average delay at each relay node. To the best of our knowledge, it is the first time in the related literature on cooperative networks with MPR capabilities, where such an analysis is performed.
\item Furthermore, for the two-sources, two-relay symmetric system, we provide an explicit expression for the average queueing delay, without the need of solving a boundary value problem.
\end{itemize}
Concluding, the analytical results in this work, to the best of our knowledge, have not been reported in the literature.

The rest of the paper is organized as follows. In Section \ref{mod} we describe the system model in detail. Section \ref{stab-anal-2} is devoted to the investigation of the throughput and the stability conditions for the asymmetric MPR model of $N=2$ sources, while in Section \ref{stab-anal-n} we generalize our previous results for the general case of $N$ sources, both with MPR capabilities at the destination. In Section \ref{anal} we focus on the delay analysis for the general asymmetric two-sources, two-relays network. A fundamental functional equation is derived, and some preparatory results in view of the resolution of the functional equation are obtained. We formulate and solve two boundary value problems, the solution of which provide the generating function of the stationary joint queue length distribution of relay nodes. The basic performance metrics are obtained, and important hints regarding their numerical evaluation are also given. In Section \ref{sym} we obtain an explicit expression for the average delay at each relay node for the symmetrical system without solving a boundary value problem, while in Section \ref{idle} we show how to adapt our analysis when we are interesting on time instants the relays get empty. Finally, numerical examples that shows insights in the system performance are given in Section \ref{num}.  

\section{Model description and notation}\label{mod}
In this work, we consider a network consisting of $N$ saturated sources, two relays, and one destination node. In order to facilitate the presentation and the description of the cooperative protocol, we describe in the following the case of $N=2$ saturated sources assisted by two relays as depicted in Fig. \ref{mood}.
\subsection{Network Model}
We consider a network of $N=2$ saturated sources, say $P_1$ and $P_2$, two relay nodes, denoted by $R_{1}$ and $R_{2}$, and a common destination node $D$ as depicted in Fig. \ref{mood}.\footnote{In this section we will present the system model for the case of two sources. However, in Section \ref{stab-anal-n}, we consider the case where there are $N$-symmetric sources.} The sources transmit packets to the node $D$ with the cooperation of the relays. The packets have equal length and the time is divided into slots corresponding to the transmission time of a packet. 

We assume that the relays and the destination have MPR capabilities and the success probabilities for the transmissions will be provided in Section \ref{sec:PHY}. As already stated, MPR is the most appropriate model to capture the wireless transmissions. 

Let $N_{i,n}$ be the number of packets in the buffer of relay node $R_{i}$, $i=1,2$, at the beginning of the $n$th slot. Moreover, during the time interval $(n,n+1]$ (i.e., during a time slot) the relay $R_{i}$, $i=1,2$ generates also packets of its own (i.e., exogenous traffic). Let $\{A_{i,n}\}_{n\in N}$ be a sequence of i.i.d. random variables where $A_{i,n}$ represents the number of packets which arrive at $R_i$ in the interval $(n,n+1]$, with $E(A_{i,n})=\widehat{\lambda}_{i}<\infty$. Under such assumptions $\mathbf{N}_{n}=\{(N_{1,n},N_{2,n});n\in\mathbb{N}\}$ is a two dimensional discrete time Markov chain with state space $\mathcal{S}=\{0,1,...\}\times\{0,1,...\}$.\footnote{The network with pure relays can be obtained by replacing $\widehat{\lambda}_{1}=\widehat{\lambda}_{2}=0$.}

The sources have random access to the medium with no coordination among them. At the beginning of a slot, the source $P_k$ attempts to transmit a packet with a probability $t_{k}$, $k=1,2$, i.e., with probability $\bar{t}_{k}=1-t_{k}$ remains silent. The sources and the relays transmit in different channel frequencies, but the destination node can overhear both of them. We assume that node $D$ first senses the sources, and if it senses no activity from them, it switches to the channel of the relays. For modeling purposes we assume that this sensing time is negligible comparing to the duration of the time slot. If a direct packet transmission from a source to node $D$ fails, and at the same time if at least one of the relays is able to decode this packet, then, the failed packet is stored in the relay buffer\footnote{See subsection \ref{relaycoop} for more details.}. The relay node is responsible to forward it to the node $D$ at a subsequent time slot. The queues at the relays are assumed to have infinite size. Note that in case node $D$ senses activity from the sources it keeps listening to the source channel during the slot, and thus it is unavailable to the relays, during that slot.
 
In case node $D$ senses no activity from the sources at the beginning of a slot, it switches to the channel of relay nodes.
Due to the interference among the relays, we consider the following state-dependent access policy: 
\begin{enumerate}
\item If both relays are non empty, $R_{i}$, $i=1,2,$ transmits a packet with probability $\alpha_{i}$ (with probability $\bar{\alpha}_{i}=1-\alpha_{i}$ remains silent).
\item If $R_{1}$ (resp. $R_{2}$) is the only non-empty, it adapts its transmission probability. More specifically, it transmits a packet with a probability $\alpha_{i}^{*}>\alpha_{i}$, in order to utilize the idle slot of the neighbor relay node (with probability $\bar{\alpha^{*}_{i}}=1-\alpha^{*}_{i}$ remains silent).\footnote{We consider the general case for $\alpha_{i}^{*}$ instead of assuming directly $\alpha_{i}^{*}=1$. This can handle cases where the node cannot transmit with probability one even if the other node is silent, e.g., when the nodes are subject to energy limitations. It is outside of the scope of this work to consider specific cases and we intent to keep the proposed analysis general. Note also that in such a case, a relay node is aware about the state of its neighbor, since in such a shared access network, it is practical to assume a minimum exchanging information of one bit between the nodes.}
\end{enumerate}  

\begin{figure}[!ht]
\centering
\includegraphics[scale=0.3]{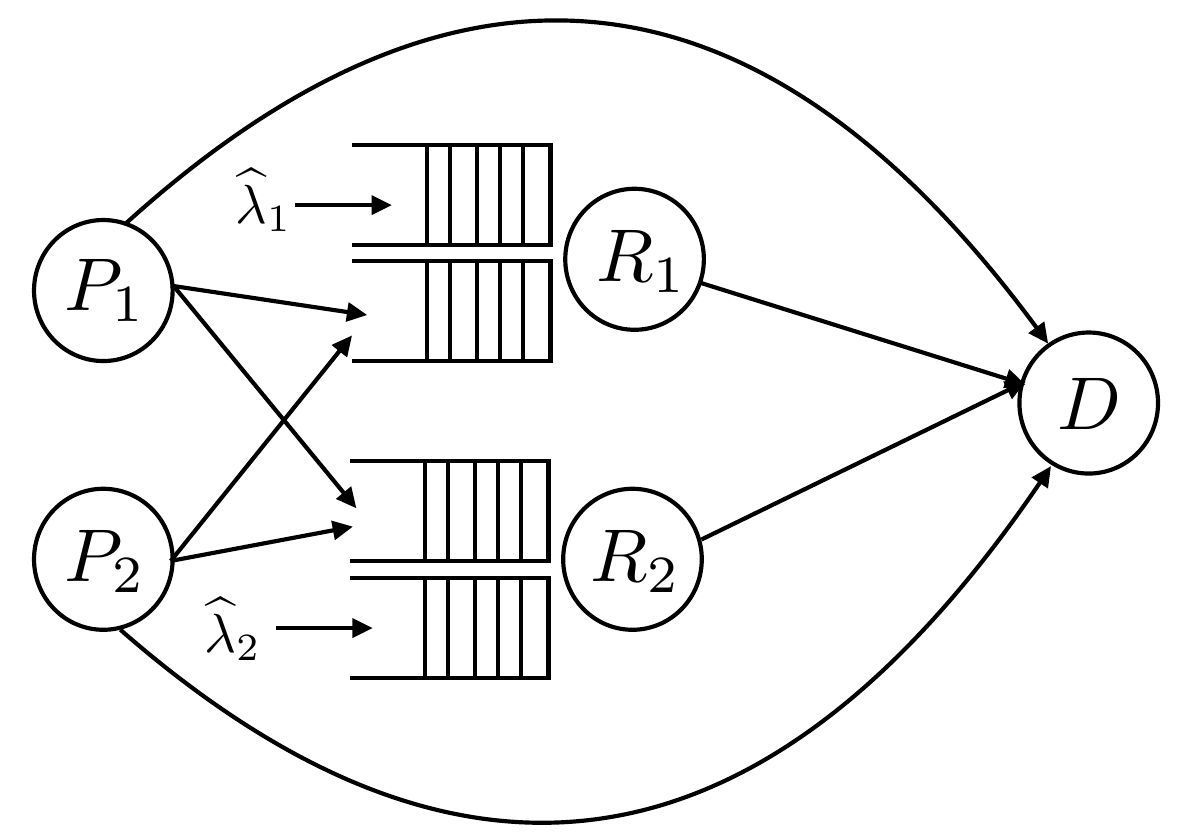}
\caption{An instance of the two-relay cooperative wireless network with two sources. In addition, the relays $R_1$ and $R_2$ have their own traffic $\widehat{\lambda}_{1}$ and $\widehat{\lambda}_{2}$ respectively, and assist the sources $P_1$ and $P_2$ by forwarding their failed packets to the destination $D$. The relays are assumed to have infinite capacity buffers.}
\label{mood}
\end{figure} 

\subsection{Description of Relay Cooperation}\label{relaycoop}
In the following, we describe the cooperation mechanism among sources and relays. As already stated, the relays overhear the direct transmission the sources to the destination and if it fails, they can store the failed packet in their queues with a specific probability, and try to forward it to the destination at a subsequent time slot. The source-relay cooperation policy is described in the following:
Denote by $p_{a_{i,j}}$, the probability that a transmitted packet from the $i$-th source will be stored at the queue of $j$-th relay if the relay is able to decode it. This probability captures two scenarios: (i) The partial cooperation of a relay, see \cite{PappasISIT2012, PappasJCN2018}, in case only one relay will be able to receive the failed packet to $D$. (ii) The scenario when both relays receive the failed packet from source $P_i$. Then, $R_{1}$ will store it in its queue with probability $p_{a_{i,1}}$, while with probability $p_{a_{i,2}}=1-p_{a_{i,1}}$ the failed packet will be stored by $R_{2}$. In this work we employ the following policy.
\begin{enumerate}
\item In order to simplify the presentation, if only one relay receives successfully a packet that fails to reach the destination, then this packet will be stored in its queue with probability 1\footnote{We would like to emphasize that in general, in case only one relay, say $R_{1}$, receives correctly a failed packet, then it will store it in its queue with probability $p_{a_{i,1}}$. This probability, controls the amount of the cooperation that this relay provides. We omit this case here for the sake of presentation, however it can be handled by our analysis trivially.}.
\item If both relays decode correctly a failed packet, then, we have the following scenarios: (i) if $p_{a_{i,1}}+p_{a_{i,2}}=1$, then the packet enters randomly either $R_{1}$, or $R_{2}$. (ii) If $p_{a_{i,1}}+p_{a_{i,2}}<1$, then there is a probability that the failed packet will not be accepted in the queues of the relays and it has to be retransmitted in a future timeslot by its source.
\end{enumerate}

\subsection{Physical Layer Model} \label{sec:PHY}
The MPR channel model used in this work is a generalized form of the packet erasure model. In wireless networks, a transmission over a link is successful with a probability. We denote $P_s(i,k,A)$ the success probability of the link between nodes $i$ and $k$ when the set of active transmitters are in $A$. For example, $\PsoneRoneonetwo$ denotes the success probability for the link between the first source and the first relay when both sources are transmitting. The probability that the transmission fails is denoted by $\PfoneRoneonetwo$. In order to take into account the interference among the relays, we have to distinguish the success probabilities when a relay transmits and the other is active or inactive (i.e., it is empty). 

Thus, when $i\in\{R_{1},R_{2}\}$, the success probability of the link between relay node $i$ and node $D$ when relay node $i$ is the only non empty is denoted by $P_s^{*}(i,D,\{i\})$. We distinguish this case, in order to have more general results that capture scenarios that one relay can increase its transmission power when the other relay is empty, thus remains silent, in order to achieve a higher success probability. Therefore we have $\PssRoneDRone\geq\PsRoneDRone$. The probabilities of successful packet reception can be obtained using the common assumption in wireless networks that a packet can be decoded correctly by the receiver if the received SINR (Signal-to-Interference-plus-Noise-Ratio) exceeds a certain threshold. The SINR depends on the modulation scheme, the target bit error rate and the number of bits in the packet \cite{tse} and the expressions for the success probabilities can be found in several papers, i.e for the case of Raleigh fading refer in \cite{PappasTWC2015}. On the other hand, if source source $P_{k},\,k=1,2$, is the only that transmits, $\PskDk$ denotes the probability that its packet is successfully decoded by the destination, while with probability $\PfkDk=1-\PskDk$ this transmission fails. 

We assume that the receivers, both at relays and at the node $D$, transmit an instantaneous error-free feedback (ACK) for all the packets that were successfully received in a slot. In case a relay receives a positive ACK from node $D$, it removes the successfully transmitted packet from its buffer. The packets that were not successfully transmitted are retained.

We now provide the service rates $\mu_{1}$, $\mu_{2}$ seen at relay nodes, given by
\begin{equation} \label{eq:mu1}
\begin{array}{rl}
\mu_1 = &\overline{t}_1 \overline{t}_2 \left[ \mathrm{Pr} (N_2 = 0) \alpha_{1}^{*} \PssRoneDRone \right.\\&\left.+ \mathrm{Pr} (N_2 > 0) \alpha_{1} \left( \alpha_{2} \PsRoneDRoneRtwo + \overline{\alpha}_{2} \PsRoneDRone \right) \right],\\
\mu_2 = &\overline{t}_1 \overline{t}_2 \left[ \mathrm{Pr} (N_1 = 0) \alpha_{2}^{*} \PssRtwoDRtwo \right.\\&\left.+ \mathrm{Pr} (N_1 > 0) \alpha_{2} \left( \alpha_{1} \PsRtwoDRoneRtwo + \overline{\alpha}_{1} \PsRtwoDRtwo \right) \right].
\end{array}
\end{equation}

Note that the success probability $\PsRoneDRoneRtwo$ (resp. $\PsRtwoDRoneRtwo$) refers to the case where a submitted packet from relay $R_{1}$ (resp. $R_{2}$) is successfully decoded by node $D$, and includes both the case where only a packet from $R_{1}$ (resp. $R_{2}$) is decoded, both the case where both relays have successful transmissions (i.e. MPR case). To simplify the presentation we define the following two variables
\begin{displaymath}
\begin{array}{rl}
\Delta_1=\PsRoneDRoneRtwo -\PsRoneDRone,\\
\Delta_2=\PsRtwoDRoneRtwo -\PsRtwoDRtwo.
\end{array}
\end{displaymath} 
These variables can be seen as an indication regarding the MPR capability for each relay. If $\Delta_i \to 0$, then the interference caused by the other relay is negligible.

\begin{table} [ht]
\caption{Basic Notation}
\label{table:notation}
	\begin{tabular}{| c | l | }
		\hline
		\bfseries Symbol & \bfseries Explanation\\
	\hline
		$N_{i,n}$ & Queue size at $R_{i}$ at the beginning of slot $n$.\\
		\hline$A_{i,n}$ &Internal packet arrivals at $R_{i}$, during $(n,n+1]$. \\
	\hline	$\widehat{\lambda}_{i}$ &Average external arrival rate at $R_{i}$, $i=1,2$.\\
		\hline$t_{k}$ &Transmission probability of $P_k$, $k=1,2$\\
\hline$a_{i}$ &Transmission probability of $R_{i}$ \\ &when both relays are non-empty.\\ 
		\hline$a^{*}_{i}$ & Transmission probability of $R_{i}$, $i=1,2$, \\&when it is the only non-empty relay.\\
	\hline	$P_{s}(k,m,A)$& Success probability of the link between nodes\\&$k$, $m$ given the set of transmitting nodes $A$.\\
		 \hline$P_{s}^{*}(R_{i},D,\{R_{i}\})$&Success probability of $R_{i}$, $i=1,2$, when \\& it is the only active (i.e., non-empty).\\
		\hline$P_{s,k}(k,R_{i},\{1,2\})$& Success probability of the link between \\& $P_k$-$R_{i}$ when both sources transmit,\\&  but source $k$ is the only successful.\\
		\hline
	\end{tabular}	
\end{table}

\section{Throughput and Stability Analysis for the two-source case -- General MPR case} \label{stab-anal-2}

In the following, we provide the throughput analysis for the two-sources case, and derive the stability conditions for the queues at the relays. Following~\cite{szpa}, our two-dimensional Markovian process $\mathbf{N}_{n}=\{(N_{1,n},N_{2,n});n\in\mathbb{N}\}$ is \emph{stable} if
$\lim_{n \rightarrow \infty} {Pr}[\mathbf{N}_{n}< \mathbf{x}] = F(\mathbf{x})$ and $\lim_{ \mathbf{x} \rightarrow \infty} F(\mathbf{x}) = 1$, where by $\mathbf{x}\to\infty$ means that $x_{i}\to\infty$, $i=1,2$, and $F(\mathbf{x})$ is the limiting distribution function. 

Loynes' theorem~\cite{Loynes} states that if the arrival and service processes of a queue are strictly jointly stationary and the average arrival rate is less than the average service rate, then the queue is stable. If the average arrival rate is greater than the average service rate, then the queue is unstable and the value of $N_{i,n}$ approaches infinity almost surely.
We proceed with the derivation of the throughput per source, which allow us to calculate the endogenous arrivals at the relays.

\subsection{Throughput per source}

Hereon we consider the throughput per source when both queues of the relays are stable. Conditions for stability are given in a subsequent subsection. When the queues at the relays are not stable the throughput per source can be obtained using the approach in \cite{PappasTWC2015}.

The throughput per source $P_k$, say $T_k$, equals the direct throughput when the transmission to the destination is successful, plus the throughput contributed by the relays (if they can decode the transmission) in case of a failed direct transmission to the destination. Thus, the throughput seen by $P_{1}$ is given by
\begin{displaymath}
T_1=T_{1,D}+T_{1,R},
\end{displaymath}
where $T_{1,D} = t_1 \bar{t}_2  \PsoneDone + t_1 t_2 \PsoneDonetwo$,  
and
\begin{equation}
\begin{array}{rl}
T_{1,R}=& t_1\bar{t}_{2}[\PfoneDone \PsoneRoneone \PfoneRtwoone\\&+\PfoneDone \PfoneRoneone \PsoneRtwoone\\
&+ \PfoneDone \PsoneRoneone \PsoneRtwoone]\\&+t_1 t_2 [\PfoneDonetwo \PsoneRoneonetwo \PfoneRtwoonetwo \\&+ \PfoneDonetwo \PfoneRoneonetwo \PsoneRtwoonetwo\\& +\PfoneDonetwo \PsoneRoneonetwo \PsoneRtwoonetwo].
\end{array}\label{eq:t1r}
\end{equation}
Similarly, we can obtain the throughout for source $P_{2}$. The aggregate, or network-wide throughput of the network when the queues at the relays are both stable is 
\begin{equation}
T_{aggr}=T_1+T_2+\widehat{\lambda}_{1}+\widehat{\lambda}_{2}.
\end{equation}

\subsection{Endogenous arrivals at the relays}

We now derive the internal (or endogenous) arrival rate from the sources to each relay. We would like to mention that the relays have also their own traffic (exogenous). Denote by $\widehat{\lambda}_{i}$, the average exogeneous arrival rate at relay $R_i$. 

A packet from a source is stored at the queue of a relay if the direct transmission to the node $D$ fails, and at the same time at least one relay decodes correctly the packet.
In the two-source case, with MPR capabilities at relays, each relay can receive up to two packets at each slot.

Denote by $\lambda_{i,j}$ the endogenous arrival (rate) probability from $i$-th source, $i=1,2$, to the queue at the $j$-th relay, $j=1,2$. Then,
\begin{equation}
\begin{array}{rl}
\lambda_{1,1} =& t_1\bar{t}_{2}[\PfoneDone \PsoneRoneone \PfoneRtwoone+\\&\PfoneDone \PsoneRoneone \PsoneRtwoone p_{a_{1,1}}] \\
&+t_1 t_2[\PfoneDonetwo \PsoneRoneonetwo \PfoneRtwoonetwo \\&+ \PfoneDonetwo \PsoneRoneonetwo \PsoneRtwoonetwo p_{a_{1,1}}],\vspace{2mm}\\
\lambda_{1,2} = &t_1 (1-t_2) [\PfoneDone \PfoneRoneone \PsoneRtwoone\\
&+\PfoneDone \PsoneRoneone \PsoneRtwoone p_{a_{1,2}}]\\
&+t_1 t_2[\PfoneDonetwo \PfoneRoneonetwo \PsoneRtwoonetwo\\&+\PfoneDonetwo \PsoneRoneonetwo \PsoneRtwoonetwo p_{a_{1,2}}].
\end{array}
\label{lambda11}
\end{equation}

Similarly we can define $\lambda_{2,1}$, and $\lambda_{2,2}$. Note that $T_{1,R} = \lambda_{1,1} + \lambda_{1,2}$, which is the relayed throughput for the source $P_{1}$ defined in (\ref{eq:t1r}).

The average arrival rate at the relay $i$ is given by
\begin{displaymath}
\lambda_{i} = \widehat{\lambda}_{i} + \lambda_{1,i} + \lambda_{2,i}.
\end{displaymath}

Recall that $p_{a_{1,1}}$ (resp. $p_{a_{1,2}}=1-p_{a_{1,1}}$) denotes the probability that the transmitted packet by the source $P_1$, which fails to reach the node $D$, and is correctly received by both relays, is finally stored at the queue of relay $R_{1}$ (resp. $R_{2}$). Thus, a packet is stored only in one queue, so we avoid wasting resources by transmitting the same packet twice. 

\subsection{Stability conditions for the queues at the relays}
We now proceed with the investigation of the stability conditions. 
The stability region of the system is defined as the set of arrival rate vectors $\boldsymbol{\lambda}=(\lambda_1, \lambda_2)$, for which the queues of the relay nodes are stable. Here, we will derive the stability analysis for the total average arrival rate at each relay, $\lambda_{i}$. 

The next theorem provides the stability criteria for the two-sources network under general MPR. Its proof is based on the principle of dominant systems developed in \cite{Rao_TIT1988,szpa}, and its main idea is to study whether the system of interest is stochastically comparable to
a simpler system that is easier to analyze; see also \cite{borst}. However, the same result can be also obtained by considering the two-dimensional Markov chain that describes our system and using the well-known Foster-Lyapunov drift criteria in \cite{fay,fay2}.
\begin{theorem}\label{Thm2users}
The stability region $\mathcal{R}$ is given by $\mathcal{R}=\mathcal{R}_1 \bigcup \mathcal{R}_2$ where
\begin{equation}
\begin{array}{ll}
\mathcal{R}_1=&\left\{ (\lambda_{1},\lambda_{2}):\lambda_{1} < \overline{t}_1 \overline{t}_2 \alpha_{1}^{*} \PssRoneDRone\right.\vspace{2mm}\\&\left. - \frac{\lambda_{2} \left[\alpha_{1}^*\PssRoneDRone - \alpha_{1} \left[\Btwo\right]\right]}{\Bone}, \right. \vspace{2mm} \\ &
\left. \lambda_2 < \overline{t}_1 \overline{t}_2 \alpha_{2} \left[  \alpha_{1} \Delta_2+ \PsRtwoDRtwo \right]
\right\rbrace,
\end{array}
\label{eq:R1_2}
\end{equation} 
and
\begin{equation}
\begin{array}{rl}
\mathcal{R}_2=& \left\lbrace (\lambda_{1},\lambda_{2}):\lambda_{2} < \overline{t}_1 \overline{t}_2 \alpha_{2}^{*} \PssRtwoDRtwo\right.\vspace{2mm}\\&\left. - \frac{\lambda_{1} \left[\alpha_{2}^*\PssRtwoDRtwo - \alpha_{2} \left[\Bfour\right]\right]}{\Bthree}, \right.  \vspace{2mm}\\& 
\left. \lambda_1 < \overline{t}_1 \overline{t}_2 \alpha_{1} \left[  \alpha_{2} \Delta_1+ \PsRoneDRone \right]
\right\rbrace.
\end{array}
\label{eq:R2_2}
\end{equation}
\end{theorem}
\begin{proof}
See \ref{proofThm2users}

\end{proof}

The stability region obtained in Theorem \ref{Thm2users} is depicted in Fig. \ref{region2}. To simplify presentation, we denote the points $A_1, A_2, B_1, B_2$ with the following expressions
\begin{eqnarray*}
A_1=\overline{t}_1 \overline{t}_2 \alpha_{1}^{*} \PssRoneDRone, \text{ }A_2=\overline{t}_1 \overline{t}_2 \alpha_{1} \left[\Btwo\right],
\\
B_1=\overline{t}_1 \overline{t}_2 \alpha_{2}^{*} \PssRtwoDRtwo, \text{ }
B_2=\overline{t}_1 \overline{t}_2  \alpha_{2} \left[\Bfour\right].
\end{eqnarray*}
\begin{figure}[!ht]
\centering
\includegraphics[scale=0.6]{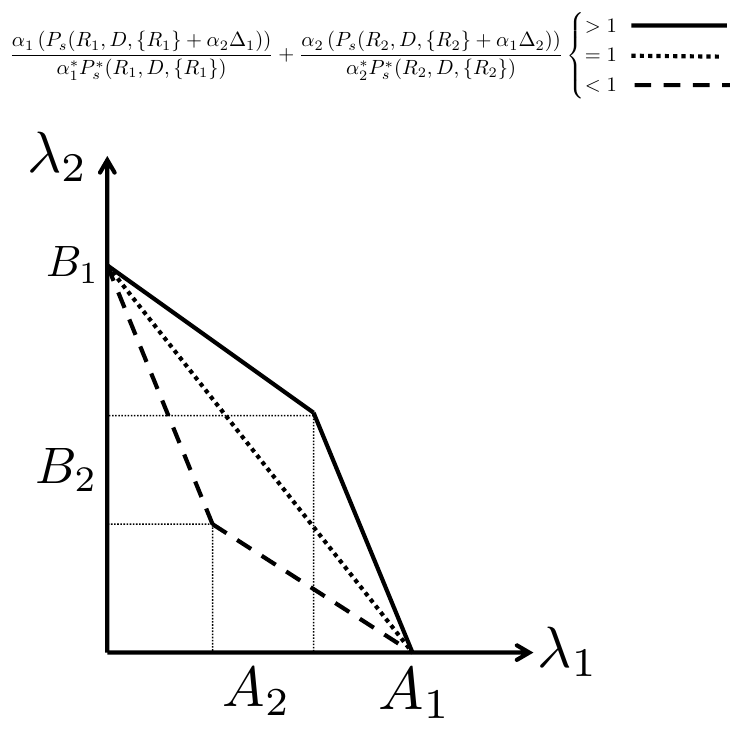}
\caption{The stability region described in Theorem \ref{Thm2users}.}
\label{region2}
\end{figure} 

\begin{remark}
The stability region is a convex polyhedron when the condition 
$\frac{\alpha_1\left(P_{s}(R_1,D,\{R_1\}+\alpha_2\Delta_1) \right)}{\alpha_{1}^{*}P_{s}^{*}(R_1,D,\{R_1\})}+\frac{\alpha_2\left(P_{s}(R_2,D,\{R_2\}+\alpha_1\Delta_2) \right)}{\alpha_{2}^{*}P_{s}^{*}(R_2,D,\{R_2\})}\geq1$ holds. 
In the previous condition, when equality holds, the region becomes a triangle and coincides with the case of time-sharing of the channel between the relays. Convexity is an important property corresponding to the case when parallel concurrent transmissions are preferable to a time-sharing scheme. Additionally, convexity of the stability region implies that if two rate pairs are stable, then any rate pair lying on the line segment joining those two rate pairs is also stable.
\end{remark}
\begin{remark}
The network without relay's assistance can be obtained by $p_{a_{1,2}}=p_{a_{1,1}}=0$. In this case, we have a network with saturated sources and also two relays with bursty traffic that transmit packets only when the saturated sources remain silent.
\end{remark}
\begin{remark}
One can connect the endogenous arrivals from the sources to the relays with the stability conditions, obtained in Theorem \ref{Thm2users}, by replacing the relevant expressions of $\widehat{\lambda}_{1}$ and $\widehat{\lambda}_{2}$ into $\lambda_{1}$ and $\lambda_{2}$.
\end{remark}

\begin{remark}
A slightly different scenario is captured by the case where the relays can transmit in a different channel than the sources and the destination can hear both channels at the same time. The receivers at the relays are operating at the same channels where the sources are transmitting. In this case, we can have a full duplex operation at the relays on different bands. Thus, we have the following average service rates for the relays
\begin{displaymath}
\begin{array}{rl}
\mu_1 =&  \mathrm{Pr} (N_2 = 0) \alpha_{1}^{*} \PsRoneDRone \\&+ \mathrm{Pr} (N_2 > 0) \alpha_{1} \left( \alpha_{2} \PsRoneDRoneRtwo + \overline{\alpha}_{2} \PsRoneDRone \right),\\
\mu_2 =&\mathrm{Pr} (N_1 = 0) \alpha_{2}^{*} \PsRtwoDRtwo \\&+ \mathrm{Pr} (N_1 > 0) \alpha_{2} \left( \alpha_{1} \PsRtwoDRoneRtwo + \overline{\alpha}_{1} \PsRtwoDRtwo \right).
\end{array}
\end{displaymath}

The stability analysis for this case can be trivially obtained by the presented analysis thus, it is omitted. However, this scenario has applicability in nowadays relay-assisted networks.
\end{remark}
\section{Throughput and Stability Analysis -- The symmetric $N$-sources case for the General MPR case} \label{stab-anal-n}
Here we will generalize the analysis provided in the previous section for the $N$-sources case. Due to presentation clarity, we focus only on the symmetric case, under which, the sources attempt to transmit with probability $t$. Moreover, the success probability of a source's transmission to the destination is the same for all the sources. Thus, in order to characterize it we just need the number of active sources, i.e. the interference. This probability is denoted by $P_s(D,i)$ to capture the case that $i$ sources are attempting transmission (including the source we intend to study its performance), similarly we define $P_s(R_j,i), j=1,2$. 

\subsection{Endogenous arrivals at the relays and throughput performance}
The direct throughput of a source to the destination in the case of $N$ symmetric sources is given by
\begin{displaymath}
T_D=\sum_{i=1}^{N} {N \choose i-1} t^i (1-t)^{N-i} P_s(D,i).
\end{displaymath}

In order to calculate the relayed throughput in the network, we have to derive the endogenous arrivals at each relay node. For the symmetric $N$-source case, we denote by $\lambda_{j,s}$ the endogenous arrivals from the sources to the relay $R_{j}$, $j=1,2$. We need further to characterize the average number of packet arrivals from the users at each relay. Denote by $r_{k,j}$ the probability that $k$ packets will arrive in a timeslot at $R_{j}$, $j=1,2$. Then, for $1 \leq k \leq N$
\begin{displaymath}
\begin{array}{rl}
r_{k,1}=&\sum_{i=k}^{N}\sum_{l=0}^{k}{N \choose i}{i \choose k}{k \choose l} {t^{i}\overline{t}^{N-i}} \left(P_s(R_1,i)\right)^{k} \left(\overline{P}_s(D,i)\right)^{k} \left(\overline{P}_s(R_2,i)\right)^{k-l}\\
&\times\left(P_s(R_2,i)\right)^{l} p_{a_1}^{l}\left[1-P_s(R_1,i)\overline{P}_s(D,i)\right]^{i-k}.
\end{array}
\end{displaymath}Similarly, we obtain $r_{k,2}$ for the second relay. Thus,
\begin{equation} \label{eq:lambdaju}
\begin{array}{c}
\lambda_{j,s} = \sum_{k=1}^{N} k r_{k,j}, \text{ }j=1,2.
\end{array}
\end{equation}

Note that the \emph{network-wide relayed throughput} when both relays are stable is given by $\lambda_{1,s}+\lambda_{2,s}$.
Thus, the aggregate or network-wide throughput of the network when both relays are stable is given by

\begin{displaymath}
T_{aggr}=N T_D + \lambda_{1,s}+\lambda_{2,s} + \widehat{\lambda}_{1}+\widehat{\lambda}_{2}.
\end{displaymath}

Recall that the total arrival rate at relay $R_{j}$ is $\lambda_{j}=\lambda_{j,s}+\widehat{\lambda}_{j}$, $j=1,2$, consisting of the endogenous arrivals from the sources and the external traffic.
\subsection{Stability conditions for the queues at the relays}
The service rates at the at the first and second relay are given by
\begin{displaymath}
\begin{array}{rl}
\mu_1 =& \overline{t}^{N} \left[ \mathrm{Pr} (N_2 = 0) \alpha_{1}^{*} \PssRoneDRone \right.\\& \left.+\mathrm{Pr} (N_2 > 0) \alpha_{1} \left( \alpha_{2} \PsRoneDRoneRtwo + \overline{\alpha}_{2} \PsRoneDRone \right) \right],\\
\mu_2 =& \overline{t}^{N} \left[ \mathrm{Pr} (N_1 = 0) \alpha_{2}^{*} \PssRtwoDRtwo \right.\\&\left.+ \mathrm{Pr} (N_1 > 0) \alpha_{2} \left( \alpha_{1} \PsRtwoDRoneRtwo + \overline{\alpha}_{1} \PsRtwoDRtwo \right) \right].
\end{array}
\end{displaymath}

Following the same methodology as in the proof of Theorem \ref{Thm2users}, we obtain the stability conditions for the symmetric $N$-sources case. The stability conditions are given by $\mathcal{R}=\mathcal{R}_1 \cup \mathcal{R}_2$ where

\begin{displaymath}
\begin{array}{rl}
\mathcal{R}_1=&\left\{ (\lambda_{1},\lambda_{2}): \lambda_{1} < \overline{t}^N \alpha_{1}^{*} \PssRoneDRone\vspace{2mm}\right.\\&\left. - \frac{\lambda_{2} \left[\alpha_{1}^*\PssRoneDRone - \alpha_{1} \left[\Btwo\right]\right]}{\Bone}, \right. \notag \\ &
\left. \lambda_2 < \overline{t}^N \alpha_{2} \left[  \alpha_{1} \Delta_2+ \PsRtwoDRtwo \right]
\right\rbrace,\vspace{2mm}\\
\mathcal{R}_2=& \left\lbrace (\lambda_{1},\lambda_{2}):\lambda_{2} < \overline{t}^N \alpha_{2}^{*} \PssRtwoDRtwo\vspace{2mm}\right.\\&\left. - \frac{\lambda_{1} \left[\alpha_{2}^*\PssRtwoDRtwo - \alpha_{2} \left[\Bfour\right]\right]}{\Bthree}, \right. \notag \\ &
\left. \lambda_1 < \overline{t}^N \alpha_{1} \left[  \alpha_{2} \Delta_1+ \PsRoneDRone \right]
\right\rbrace.
\end{array}
\end{displaymath}

\section{Delay analysis: The two-sources case}\label{anal}
This section is devoted to the analysis of the queueing delay experienced at the relays. Our aim is to obtain the generating function of the joint stationary distribution of the number of packets at relay nodes. In the following we consider the case of $N=2$ sources, and focus on a subclass of MPR models, called the ``capture" channel, under which at most one packet can be successfully decoded by the receiver of the node $D$, even if more than one nodes transmit\footnote{Recall that the assumption of two sources is not restrictive, and our analysis can be extended to the general case of $N$ sources. Furthermore, our analysis remains valid even for the case of general MPR model. However, both cases require some additional essential technical results, which in turn will worse the readability of the paper and further increase its length. Besides, recent advances in LoRaWAN \cite{bankov} reveals the applicability of the capture channel model.}.

In order to proceed, we have to provide more information regarding the success probabilities of a transmission between nodes that were defined in subsection \ref{sec:PHY}. More precisely, we have to take into account the number as well as the type of nodes that transmit (i.e., a source or a relay node). This is due to several reasons, such as the fact that generally the channel quality between relay nodes and destination is usually better than the channel between sources and destination, as well as due to the wireless interference, since the channel quality is severely affected by the the number of nodes that attempt to transmit. Moreover, it is crucial to take into account the possibility that a failed packet can be successfully decoded by both relays, as well as the ability of our ``smart" relay nodes to be aware of the status of the others, which in turn leads to self-aware networks. With that in mind we consider the following cases:
\begin{enumerate}
\item Both sources transmit
\begin{enumerate}
\item When both sources \textit{fail} to transmit directly to the node $D$, the failed packet of source $k$ is successfully decoded by relay $R_{i}$ with probability $\PskRionetwo$, $k=1,2$, $i=1,2$, where with probability $\PfRionetwo$, the relay $R_{i}$ failed to decode both packets. Note also that $\overline{P}_{s}(1,R_{i},\{1,2\})=\overline{P}_{s}(0,R_{i},\{1,2\})+P_{s}(2,R_{i},\{1,2\})$, $\overline{P}_{s}(2,R_{i},\{1,2\})=\overline{P}_{s}(0,R_{i},\{1,2\})+P_{s}(1,R_{i},\{1,2\})$, $i=1,2$. Due to the total probability law we have 
\begin{displaymath}
\begin{array}{l}
(\PfRoneonetwo+\PsoneRoneonetwo+\PstwoRoneonetwo)\\\times(\PfRtwoonetwo+\PsoneRtwoonetwo+\PstwoRtwoonetwo)=1.
\end{array}
\end{displaymath}
\item When source 1 (resp. source 2) is the only that succeeds to transmit a packet at node $D$, i.e., its transmission was successfully decoded by node $D$, then with probability $P_{s,2}(2,R_{i},\{1,2\})$ (resp. $P_{s,1}(1,R_{i},\{1,2\})$), $i=1,2$, the failed packet of source $2$ (resp. source $1$) is successfully decoded by the relay $R_{i}$. On the contrary, with probability $\overline{P}_{s,2}(2,R_{i},\{1,2\})$ (resp. $\overline{P}_{s,1}(1,R_{i},\{1,2\})$), the relay $R_{i}$ failed to decode the packet from source $2$ (resp. source $1$), and thus, the packet must be retransmitted by its source. Due to the total probability law we have,
\begin{displaymath}
\begin{array}{rl}
(P_{s,2}(2,R_{1},\{1,2\})+\overline{P}_{s,2}(2,R_{1},\{1,2\}))\\\times(P_{s,2}(2,R_{2},\{1,2\})+\overline{P}_{s,2}(2,R_{2},\{1,2\}))&=1,\\
(P_{s,1}(1,R_{1},\{1,2\})+\overline{P}_{s,1}(1,R_{1},\{1,2\}))\\\times(P_{s,1}(1,R_{2},\{1,2\})+\overline{P}_{s,1}(1,R_{2},\{1,2\}))&=1.
\end{array}
\end{displaymath}
\end{enumerate}
\item Only one source transmit, say source $k$, and the other remains silent. When source $k$ fails to transmit directly to node $D$, its failed packet is successfully decoded by relay $R_{i}$ with probability $\PskRik$, $k=1,2$, $i=1,2$, where with probability $\PfRik$, the relay $R_{i}$ fails to decode the packet.
Due to the total probability law we have 
\begin{displaymath}
(\PfRonek+\PskRonek)(\PfRtwok+\PskRtwok)=1.
\end{displaymath}
\end{enumerate}
Note that the cases $1. b)$ and $2.$ refer to the case where only one source is able to cooperate with relays. We distinguished these two cases because in the former one, there is an interaction among sources since both of them transmit, while in the latter one, only one source transmit and the other remains silent (i.e., there is no interaction). Such an interaction, plays a crucial role on the values of the success probabilities. In wireless systems, the interference and interaction among transmitting nodes is of great importance and have to be taken into account.

If both relays transmit simultaneously, with probability $\PsRiDRoneRtwo$, the packet transmitted from $R_{i}$ is successfully received by node $D$, while with probability $\PfRiDRoneRtwo=1-\sum_{i=1,2}\PsRiDRoneRtwo$, both of them failed to be received by the node $D$, and have to be retransmitted in a later time slot. Recall also the success probabilities $\PssRiDRi$, $\PsRiDRi$ of $R_i$ when the other relay node is active (i.e., non-empty), and inactive (i.e., empty) respectively. We assume that $\PssRiDRi>\PsRiDRi>\PsRiDRoneRtwo$. Denote the counter probabilities $\PfsRiDRi=1-\PssRiDRi$, $\PfRiDRi=1-\PsRiDRi$, $i=1,2$.

In the following, we proceed with the derivation of a functional equation, the solution of which, provides the generating function of the stationary joint queue length distribution at relay nodes. This result is the key element for obtaining expressions for the queueing delay at relay nodes.

\subsection{Functional equation and preparatory results}
Clearly, $\mathbf{N}_{n}=\{(N_{1,n},N_{2,n};n\in\mathbb{N})\}$ is a discrete time Markov chain with state space $\mathcal{S}=\{(k_{1},k_{2}):k_{1},k_{2}=0,1,2,...\}$. The queues of both relay nodes evolve as:
\begin{equation}
\begin{array}{c}
N_{i,n+1}=[N_{i,n}+F_{i,n}]^{+}+A_{i,n},\,i=1,2,
\end{array}
\label{x}
\end{equation}
where $F_{i,n}$ is either the number of arrivals (in such a case $F_{i,n}$ equals 0 or 1) at relay $R_{i}$, or the number of departures (in such a case $F_{i,n}$ equals $0$ or $-1$) from $R_i$ at time slot $n$. In the former case either both sources transmit simultaneously, and the unsuccessful packet is stored in $R_{i}$, or only a single source transmits, but its transmission was unsuccessful. The latter case occurs when the sources remain silent and $R_i$ attempts to transmit a packet at node $D$. 

Recall that the relays have also their own traffic, and $A_{i,n}$ represents the number of arrivals (of such generated traffic) in the the time interval $(n,n+1]$. Let $H(x,y)$ be the generating function of the joint stationary queue process and $Z(x,y)$ the generating function of the joint distribution of the number of arriving packets in any slot (i.e., self-generated traffic of the relays), viz.
\begin{displaymath}
\begin{array}{c}
H(x,y)=\lim_{n\to\infty}E(x^{N_{1,n}}y^{N_{2,n}}),\,|x|\leq1,|y|\leq1,\\
Z(x,y)=\lim_{n\to\infty}E(x^{A_{1,n}}y^{A_{2,n}}),\,|x|\leq1,|y|\leq1.
\end{array}
\end{displaymath}

In the following we assume for sake of convenience only a particular distribution for the self-generated arrival processes at both relays, namely the geometric distribution\footnote{Note that such a distribution is natural in radio-packet networks. However, our analysis remains valid, with some modifications, even for a more general arrival process.} \cite{nain}. We further assume that both arrival processes are independent. Thus,
\begin{displaymath}
\begin{array}{c}
Z(x,y)=[(1+\widehat{\lambda}_{1}(1-x))(1+\widehat{\lambda}_{2}(1-y))]^{-1}.
\end{array}
\end{displaymath}
By exploiting (\ref{x}), and using (\ref{fc}), we obtain after lengthy calculations 
\begin{equation}
\begin{array}{c}
R(x,y)H(x,y)=A(x,y)H(x,0)+B(x,y)H(0,y)+C(x,y)H(0,0),
\end{array}
\label{we}
\end{equation}
where,
\begin{displaymath}
\begin{array}{rl}
R(x,y)=&Z^{-1}(x,y)-1+\bar{t}_{1}\bar{t}_{2}[\alpha_{1}\widehat{\alpha}_{2}(1-\frac{1}{x})+\alpha_{2}\widehat{\alpha}_{1}(1-\frac{1}{y})]\\&+(1-x)L_{1}+(1-y)L_{2}+(1-xy)L_{3},
\end{array}
\end{displaymath}
and
\begin{displaymath}
\begin{array}{rl}
L_{1}=& t_{1}\bar{t}_{2}\PfoneDone\PfoneRtwoone\PsoneRoneone\\&+t_{2}\bar{t}_{1}\PftwoDtwo\PftwoRtwotwo \PstwoRonetwo\\&+t_{1}t_{2}[\PfDonetwo\PfRtwoonetwo(\PsoneRoneonetwo+\PstwoRoneonetwo)\\
&+\PsoneDonetwo\overline{P}_{s,2}(2,R_{2},\{1,2\})P_{s,2}(2,R_{1},\{1,2\})\\&+\PstwoDonetwo\overline{P}_{s,1}(1,R_{2},\{1,2\})P_{s,1}(1,R_{1},\{1,2\})],
\end{array}
\end{displaymath}
\begin{displaymath}
\begin{array}{rl}
L_{2}=& t_{1}\bar{t}_{2}\PfoneDone\PsoneRtwoone\PfoneRoneone\\&+t_{2}\bar{t}_{1}\PftwoDtwo\PstwoRtwotwo \PftwoRonetwo\\&+t_{1}t_{2}[\PfDonetwo\PfRoneonetwo(\PsoneRtwoonetwo+\PstwoRtwoonetwo)\\
&+\PsoneDonetwo\overline{P}_{s,2}(2,R_{1},\{1,2\})P_{s,2}(2,R_{2},\{1,2\})\\&+\PstwoDonetwo\overline{P}_{s,1}(1,R_{1},\{1,2\})P_{s,1}(1,R_{2},\{1,2\})],
\end{array}
\end{displaymath}

\begin{displaymath}
\begin{array}{rl}
L_{3}=& t_{1}\bar{t}_{2}\PfoneDone\PsoneRtwoone\PsoneRoneone\\&+t_{2}\bar{t}_{1}\PftwoDtwo\PstwoRtwotwo \PstwoRonetwo\\&+t_{1}t_{2}[\PfDonetwo(\PsoneRtwoonetwo+\PstwoRtwoonetwo)\\&\times(\PsoneRoneonetwo+\PstwoRoneonetwo)\\
&+\PsoneDonetwo P_{s,2}(2,R_{1},\{1,2\})P_{s,2}(2,R_{2},\{1,2\})\\&+\PstwoDonetwo P_{s,1}(1,R_{1},\{1,2\})P_{s,1}(1,R_{2},\{1,2\})],
\end{array}
\end{displaymath}

\begin{displaymath}
\begin{array}{rl}
A(x,y)=&\bar{t}_{1}\bar{t}_{2}[d_{1}(1-\frac{1}{x})+\alpha_{2}\widehat{\alpha}_{1}(1-\frac{1}{y})],\\
B(x,y)=&\bar{t}_{1}\bar{t}_{2}[d_{2}(1-\frac{1}{y})+\alpha_{1}\widehat{\alpha}_{2}(1-\frac{1}{x})],\\
C(x,y)=&\bar{t}_{1}\bar{t}_{2}[d_{1}(\frac{1}{x}-1)+d_{2}(\frac{1}{y}-1)],
\end{array}
\end{displaymath}
\begin{displaymath}
\begin{array}{rl}
\widehat{\alpha}_{i}=&\bar{\alpha}_{i}\PsRiDRi+\alpha_{i}\PsRiDRoneRtwo,\,i=1,2,\\
d_{1}=&\alpha_{1}\widehat{\alpha}_{2}-\alpha_{1}^{*}\PssRoneDRone,\\
d_{2}=&\alpha_{2}\widehat{\alpha}_{1}-\alpha_{2}^{*}\PssRtwoDRtwo.
\end{array}
\end{displaymath}
\begin{remark}
Note that $L_{i}$, $i=1,2,3$, has a clear probabilistic interpretation. Indeed, $L_{1}$ (resp. $L_{2}$) is the probability that a (failed) transmitted source packet will be decoded and stored at relay $R_{i}$. Moreover, $L_{3}$ is the probability that a failed transmitted source packet will be decoded and stored at both relays.
\end{remark}

Some interesting relations can be obtained directly from (\ref{we}). Taking $y = 1$, dividing by $x-1$ and taking $x\to 1$ in (\ref{we}) and vice versa yield the following ``conservation of flow" relations:
\begin{equation}
\lambda_{1}=\bar{t}_{1}\bar{t}_{2}\{\alpha_{1}\widehat{\alpha}_{2}(1-H(0,1))-d_{1}(H(1,0)-H(0,0))\},
\label{r1}
\end{equation}
\begin{equation}
\lambda_{2}=\bar{t}_{1}\bar{t}_{2}\{\alpha_{2}\widehat{\alpha}_{1}(1-H(1,0))-d_{2}(H(0,1)-H(0,0))\},
\label{r2}
\end{equation}
where for $i=1,2,$
\begin{displaymath}
\begin{array}{l}
\lambda_{i}=\widehat{\lambda}_{i}+\lambda_{1,i}+\lambda_{2,i},
\end{array}
\end{displaymath}
and,
\begin{displaymath}
\begin{array}{rl}
\lambda_{1,1}=&t_{1}\bar{t}_{2}\PfoneDone\PsoneRoneone(\PsoneRtwoone+\PfoneRtwoone)\\
&+t_{1}t_{2}[\PfDonetwo(\PfoneRtwoonetwo+\PsoneRtwoonetwo)\PsoneRoneonetwo\\
&+\PstwoDonetwo(\overline{P}_{s,1}(1,R_{2},\{1,2\})+P_{s,1}(1,R_{2},\{1,2\}))P_{s,1}(1,R_{1},\{1,2\})],\vspace{2mm}\\
\lambda_{2,1}=&t_{2}\bar{t}_{1}\PftwoDtwo\PstwoRonetwo(\PstwoRtwotwo+\PftwoRtwotwo)\\
&+t_{1}t_{2}[\PfDonetwo(\PftwoRtwoonetwo+\PstwoRtwoonetwo)\PstwoRoneonetwo\\
&+\PsoneDonetwo(\overline{P}_{s,2}(2,R_{2},\{1,2\})+P_{s,2}(2,R_{2},\{1,2\}))P_{s,2}(2,R_{1},\{1,2\})],
\end{array}
\end{displaymath}
\begin{displaymath}
\begin{array}{rl}
\lambda_{1,2}=&t_{1}\bar{t}_{2}\PfoneDone\PsoneRtwoone(\PsoneRoneone+\PfoneRoneone)\\
&+t_{1}t_{2}[\PfDonetwo(\PfoneRoneonetwo+\PsoneRoneonetwo)\PsoneRtwoonetwo\\
&+\PstwoDonetwo(\overline{P}_{s,1}(1,R_{1},\{1,2\})+P_{s,1}(1,R_{1},\{1,2\}))P_{s,1}(1,R_{2},\{1,2\})],\vspace{2mm}\\
\lambda_{2,2}=&t_{2}\bar{t}_{1}\PftwoDtwo\PstwoRtwotwo(\PstwoRonetwo+\PftwoRonetwo)\\
&+t_{1}t_{2}[\PfDonetwo(\PfoneRoneonetwo+\PsoneRoneonetwo)\PstwoRtwoonetwo\\
&+\PsoneDonetwo(\overline{P}_{s,2}(2,R_{1},\{1,2\})+P_{s,2}(2,R_{1},\{1,2\}))P_{s,2}(2,R_{2},\{1,2\})].
\end{array}
\end{displaymath}
From (\ref{r1}), (\ref{r2}) we realize that the analysis is distinguished in two cases:
\begin{enumerate}
\item For $\frac{\alpha_{1}\widehat{\alpha}_{2}}{\alpha_{1}^{*}\PssRoneDRone}+\frac{\alpha_{2}\widehat{\alpha}_{1}}{\alpha_{2}^{*}\PssRtwoDRtwo}=1$, (\ref{r1}), (\ref{r2}) yield
\begin{displaymath}
\begin{array}{c}
H(0,0)=1-\frac{\lambda_{1}}{\bar{t}_{1}\bar{t}_{2}\alpha_{1}^{*}\PssRoneDRone}-\frac{\lambda_{2}}{\bar{t}_{1}\bar{t}_{2}\alpha_{2}^{*}\PssRtwoDRtwo}=1-\rho.
\end{array}
\end{displaymath}
\item For $\frac{\alpha_{1}\widehat{\alpha}_{2}}{\alpha_{1}^{*}\PssRoneDRone}+\frac{\alpha_{2}\widehat{\alpha}_{1}}{\alpha_{2}^{*}\PssRtwoDRtwo}\neq1$, (\ref{r1}), (\ref{r2}) yield
\begin{equation}
\begin{array}{l}
H(1,0)=\frac{d_{2}\lambda_{1}+\alpha_{1}\widehat{\alpha}_{2}(\bar{t}_{1}\bar{t}_{2}\alpha_{2}^{*}\PssRtwoDRtwo-\lambda_{2})+d_{2}\alpha_{1}^{*}\PssRoneDRone H(0,0)}{\bar{t}_{1}\bar{t}_{2}(\alpha_{1}\widehat{\alpha}_{2}\alpha_{2}\widehat{\alpha}_{1}-d_{1}d_{2})},\vspace{2mm}\\
H(0,1)=\frac{d_{1}\lambda_{2}+\alpha_{2}\widehat{\alpha}_{1}(\bar{t}_{1}\bar{t}_{2}\alpha_{1}^{*}\PssRoneDRone-\lambda_{1})+d_{1}\alpha_{2}^{*}\PssRtwoDRtwo H(0,0)}{\bar{t}_{1}\bar{t}_{2}(\alpha_{1}\widehat{\alpha}_{2}\alpha_{2}\widehat{\alpha}_{1}-d_{1}d_{2})}.
\end{array}
\label{rd}
\end{equation}
\end{enumerate}

The key element to investigate the queueing delay at relay nodes is to solve the functional equation (\ref{we}) and obtain $H(x,y)$. The solution of (\ref{we}) requires first to obtain the boundary functions $H(x,0)$, $H(0,y)$ and the term $H(0,0)$. The theory of boundary value problems \cite{coh,fay} is our basic methodological tool to accomplish our goal. Since we are dealing with a quite technical approach we summarized in the following the basic steps.
\begin{enumerate}
\item[Step 1]Using (\ref{we}), we show that $H(x,0)$, $H(0,y)$ satisfy certain boundary value problems of Riemann-Hilbert-Carleman type \cite{fay}, with boundary conditions on closed curves. Lemma \ref{SQ} provides information about these curves. Its proof (see \ref{a1}) requires the investigation of the kernel $R(x,y)$ (see subsection \ref{ker}, and Lemmas \ref{LEM}, \ref{lem1}; the proof of Lemma \ref{LEM} is given in \ref{a0}). Note that based on the values of the system parameters, the unit disc may lie inside the region bounded by these contours. Clearly, the functions $H(x,0)$, $H(0,y)$ are analytic inside the unit disc, but they might have poles in the region bounded by the unit disc and these closed curves. The position of these poles (if exist) are investigated in \ref{ap4}. With that in mind, $H(x,0)$, $H(0,y)$ admit analytic continuations in the whole interiors of the curves; see also Chapter 3 in \cite{fay}. Then, we proceed with the derivation of the boundary conditions on these curves; see (\ref{p1}), (\ref{df3}) respectively.
\item[Step 2] The next step is to transform these problems into boundary value problems of Riemann-Hilbert type on the unit disc; see (\ref{zx}). This conversion is motivated by the fact that the latter problems are more usual and by far more treated in the literature \cite{coh}.
\item[Step 3] Finally, we solve these new problems by providing an integral expression of the unknown boundary functions; see (\ref{soll}) and (\ref{fin}).
\end{enumerate}
\subsection{Analysis of the kernel}\label{ker}
In the following we focus on the kernel $R(x,y)$, and provide some important properties for the following analysis. To the best of our knowledge, this type of kernel has never been treated in the related literature. Clearly,
\begin{displaymath}
R(x,y)=a(x)y^{2}+b(x)y+c(x)=\widehat{a}(y)x^{2}+\widehat{b}(y)x+\widehat{c}(y),
\end{displaymath}
where
\begin{displaymath}
\begin{array}{rl}
a(x)=&-x[\widehat{\lambda}_{2}(1+\widehat{\lambda}_{1}(1-x))+L_{2}+L_{3}x],\\
b(x)=&x[\widehat{\lambda}+\widehat{\lambda}_{1}\widehat{\lambda}_{2}+\bar{t}_{1}\bar{t}_{2}(\alpha_{1}\widehat{\alpha}_{2}+\alpha_{2}\widehat{\alpha}_{1})+L_{1}+L_{2}+L_{3}]-\bar{t}_{1}\bar{t}_{2}\alpha_{1}\widehat{\alpha}_{2}\\&-[\widehat{\lambda}_{1}(1+\widehat{\lambda}_{2})+L_{1}]x^{2},\\
c(x)=&-\bar{t}_{1}\bar{t}_{2}\alpha_{2}\widehat{\alpha}_{1}x,
\end{array}
\end{displaymath}
\begin{displaymath}
\begin{array}{rl}
\widehat{a}(y)=&-y[\widehat{\lambda}_{1}(1+\widehat{\lambda}_{2}(1-y))+L_{1}+L_{3}y],\\
\widehat{b}(y)=&y[\widehat{\lambda}+\widehat{\lambda}_{1}\widehat{\lambda}_{2}+\bar{t}_{1}\bar{t}_{2}(\alpha_{1}\widehat{\alpha}_{2}+\alpha_{2}\widehat{\alpha}_{1})+L_{1}+L_{2}+L_{3}]-\bar{t}_{1}\bar{t}_{2}\alpha_{2}\widehat{\alpha}_{1}\\&-[\widehat{\lambda}_{2}(1+\widehat{\lambda}_{1})+L_{2}]y^{2},\\
\widehat{c}(y)=&-\bar{t}_{1}\bar{t}_{2}\alpha_{1}\widehat{\alpha}_{2}y.
\end{array}
\end{displaymath}
The roots of $R(x,y)=0$ are $X_{\pm}(y)=\frac{-\widehat{b}(y)\pm\sqrt{D_{y}(y)}}{2\widehat{a}(y)}$, $Y_{\pm}(x)=\frac{-b(x)\pm\sqrt{D_{x}(x)}}{2a(x)}$, where $D_{y}(y)=\widehat{b}(y)^{2}-4\widehat{a}(y)\widehat{c}(y)$, $D_{x}(x)=b(x)^{2}-4a(x)c(x)$.

\begin{lemma}\label{LEM}
For $|y|=1$, $y\neq1$, the kernel equation $R(x,y)=0$ has exactly one root $x=X_{0}(y)$ such that $|X_{0}(y)|<1$. For $\lambda_{1}<\bar{t}_{1}\bar{t}_{2}\alpha_{1}\widehat{\alpha}_{2}$, $X_{0}(1)=1$. Similarly, we can prove that $R(x,y)=0$ has exactly one root $y=Y_{0}(x)$, such that $|Y_{0}(x)|\leq1$, for $|x|=1$.
\end{lemma}
\begin{proof}
See \ref{a0}.
\end{proof}
Using simple algebraic arguments, the following lemma provides information about the location of the branch points of the two-valued functions $Y(x)$, $X(y)$.
\begin{lemma}\label{lem1}
The algebraic function $Y(x)$, defined by $R(x,Y(x)) = 0$, has four real branch points $0< x_{1}<x_{2}\leq1<x_{3}<x_{4}<\frac{1+\tilde{\lambda}_{1}}{\tilde{\lambda}_{1}}$. Moreover, $D_{x}(x)<0$, $x\in(x_{1},x_{2})\cup(x_{3},x_{4})$ and $D_{x}(x)>0$, $x\in(-\infty,x_{1})\cup(x_{2},x_{3})\cup(x_{4},\infty)$. Similarly, $X(y)$, defined by $R(X(y),y) = 0$, has four real branch points $0\leq y_{1}<y_{2}\leq1<y_{3}<y_{4}<\frac{1+\tilde{\lambda}_{2}}{\tilde{\lambda}_{2}}$, and $D_{x}(y)<0$, $y\in(y_{1},y_{2})\cup(y_{3},y_{4})<$ and $D_{x}(y)>0$, $y\in(-\infty,y_{1})\cup(y_{2},y_{3})\cup(y_{4},\infty)$. 
\end{lemma}
To ensure the continuity of $Y(x)$ (resp. $X(y)$) we consider the following cut planes: $\tilde{C}_{x}=\mathbb{C}_{x}-([x_{1},x_{2}]\cup[x_{3},x_{4}]$, $\tilde{C}_{y}=\mathbb{C}_{y}-([y_{1},y_{2}]\cup[y_{3},y_{4}]$, where $\mathbb{C}_{x}$, $\mathbb{C}_{y}$ the complex planes of $x$, $y$, respectively. In $\tilde{C}_{x}$ (resp. $\tilde{C}_{y}$), denote by $Y_{0}(x)$ (resp. $X_{0}(y)$) the zero of $R(x,Y(x))=0$ (resp. $R(X(y),y)=0$) with the smallest modulus, and $Y_{1}(x)$ (resp. $X_{1}(y)$) the other one. 

Define the image contours, $\mathcal{L}=Y_{0}[\overrightarrow{\underleftarrow{x_{1},x_{2}}}]$, $\mathcal{M}=X_{0}[\overrightarrow{\underleftarrow{y_{1},y_{2}}}]$, where $[\overrightarrow{\underleftarrow{u,v}}]$ stands for the contour traversed from $u$ to $v$ along the upper edge of the slit $[u,v]$ and then back to $u$ along the lower edge of the slit. The following lemma shows that the mappings $Y(x)$, $X(y)$, for $x\in[x_{1},x_{2}]$, $y\in[y_{1},y_{2}]$ respectively, give rise to the smooth and closed contours $\mathcal{L}$, $\mathcal{M}$ respectively:
\begin{lemma}\label{SQ}\begin{enumerate}\item For $y\in[y_{1},y_{2}]$, the algebraic function $X(y)$ lies on a closed contour $\mathcal{M}$, which is symmetric with respect to the real line and defined by
\begin{displaymath}
\begin{array}{l}
|x|^{2}=m(Re(x)),\,m(\delta)=\frac{\widehat{c}(\zeta(\delta))}{\widehat{a}(\zeta(\delta))},\,|x|^{2}\leq\frac{\widehat{c}(y_{2})}{\widehat{a}(y_{2})},
\end{array}
\end{displaymath}
where, $\zeta(\delta)=\frac{k_{2}(\delta)-\sqrt{k_{2}^{2}(\delta)-4\bar{t}_{1}\bar{t}_{2}\alpha_{2}\widehat{\alpha}_{1}k_{1}(\delta)}}{2k_{1}(\delta)}$,
\begin{displaymath}
\begin{array}{rl}
k_{1}(\delta):=&\widehat{\lambda}_{2}(1+\widehat{\lambda}_{1})+L_{2}+2\delta(L_{3}-\widehat{\lambda}_{1}\widehat{\lambda}_{2}),\\
k_{2}(\delta):=&(1+2\delta)(L_{1}+\widehat{\lambda}_{1}(1+\widehat{\lambda}_{2}))+\widehat{\lambda}_{2}+L_{2}+L_{3}+\bar{t}_{1}\bar{t}_{2}(\alpha_{1}\widehat{\alpha}_{2}+\alpha_{2}\widehat{\alpha}_{1}).
\end{array}
\end{displaymath}
Set $\beta_{0}:=\sqrt{\frac{\widehat{c}(y_{2})}{\widehat{a}(y_{2})}}$, $\beta_{1}:=-\sqrt{\frac{\widehat{c}(y_{1})}{\widehat{a}(y_{1})}}$ the extreme right and left point of $\mathcal{M}$, respectively.
\item For $x\in[x_{1},x_{2}]$, the algebraic function $Y(x)$ lies on a closed contour $\mathcal{L}$, which is symmetric with respect to the real line and defined by
\begin{displaymath}
\begin{array}{l}
|y|^{2}=v(Re(y)),\,v(\delta)=\frac{c(\theta(\delta))}{a(\theta(\delta))},\,|y|^{2}\leq\frac{c(x_{2})}{a(x_{2})},
\end{array}
\end{displaymath}
where $\theta(\delta)=\frac{l_{2}(\delta)-\sqrt{l_{2}^{2}(\delta)-4\bar{t}_{1}\bar{t}_{2}\alpha_{1}\widehat{\alpha}_{2}l_{1}(\delta)}}{2l_{1}(\delta)}$,
\begin{displaymath}
\begin{array}{rl}
l_{1}(\delta):=&\widehat{\lambda}_{1}(1+\widehat{\lambda}_{2})+L_{1}+2\delta(L_{3}-\widehat{\lambda}_{1}\widehat{\lambda}_{2}),\\
l_{2}(\delta):=&(1+2\delta)(L_{2}+\widehat{\lambda}_{2}(1+\widehat{\lambda}_{1}))+\widehat{\lambda}_{1}+L_{2}+L_{3}+\bar{t}_{1}\bar{t}_{2}(\alpha_{1}\widehat{\alpha}_{2}+\alpha_{2}\widehat{\alpha}_{1}).
\end{array}
\end{displaymath}
Set $\eta_{0}:=\sqrt{\frac{c(x_{2})}{a(x_{2})}}$, $\eta_{1}=-\sqrt{\frac{c(x_{1})}{a(x_{1})}}$, the extreme right and left point of $\mathcal{L}$, respectively.
\end{enumerate}
\end{lemma}
\begin{proof} See \ref{a1}.
\end{proof}

\subsection{The boundary value problems}\label{bound}
As indicated in the previous section the analysis is distinguished in two cases, which differ both from the modeling and the technical point of view.
\subsubsection{A Dirichlet boundary value problem}\label{dir}
Consider the case $\frac{\alpha_{1}\widehat{\alpha}_{2}}{\alpha_{1}^{*}\PssRoneDRone}+\frac{\alpha_{2}\widehat{\alpha}_{1}}{\alpha_{2}^{*}\PssRtwoDRtwo}=1$. Then, 
\begin{displaymath}
\begin{array}{c}
A(x,y)=\frac{d_{1}}{\alpha_{1}\widehat{\alpha}_{2}}B(x,y)\Leftrightarrow A(x,y)=\frac{\alpha_{2}\widehat{\alpha}_{1}}{d_{2}}B(x,y).
\end{array}
\end{displaymath}
%Therefore, for $y\in \mathcal{D}_{y}=\{x\in\mathcal{C}:|y|\leq1,|X_{0}(y)|\leq1\}$,
%\begin{equation}
%\begin{array}{l}
%\alpha_{2}\widehat{\alpha}_{1}H(X_{0}(y),0)+d_{2}H(0,y)+\frac{\alpha_{2}\widehat{\alpha}_{1}(1-\rho)C(X_{0}(y),y)}{A(X_{0}(y),y)}=0.
%\end{array}
%\label{con}
%\end{equation}
The following theorem summarizes the main result of this subsection.
\begin{theorem}\label{dirbv}
For $\rho<1$, $H(x,0)$ is derived as the solution of a Dirichlet boundary value problem on $\mathcal{M}$, given by
\begin{equation}
\begin{array}{c}
H(x,0)=(1-\rho)\{1+\frac{2\gamma(x)i}{\pi}\int_{0}^{\pi}\frac{f(e^{i\phi})\sin(\phi)}{1-2\gamma(x)\cos(\phi)-\gamma(x)^{2}}d\phi\},\,x\in G_{\mathcal{M}},
\end{array}
\label{soll}
\end{equation}
where $G_{\mathcal{M}}$ is the interior domain bounded by the closed contour $\mathcal{M}$, and $\gamma(.)$ a conformal mapping, see \ref{conf}. A similar integral expression can be derived for $H(0,y)$ by solving another Dirichlet boundary value problem on $\mathcal{L}$. Then, using (\ref{we}) we uniquely obtain $H(x,y)$.
\end{theorem}
\begin{proof}
See \ref{diri}
\end{proof}
\subsubsection{A homogeneous Riemann-Hilbert boundary value problem}\label{rh}
In case $\frac{\alpha_{1}\widehat{\alpha}_{2}}{\alpha_{1}^{*}\PssRoneDRone}+\frac{\alpha_{2}\widehat{\alpha}_{1}}{\alpha_{2}^{*}\PssRtwoDRtwo}\neq1$,
consider the following transformation:
\begin{displaymath}
\begin{array}{rl}
G(x):=&H(x,0)+\frac{\alpha_{1}^{*}\PssRoneDRone d_{2}H(0,0)}{d_{1}d_{2}-\alpha_{1}\widehat{\alpha}_{2}\alpha_{2}\widehat{\alpha}_{1}},\\
L(y):=&H(0,y)+\frac{\alpha_{2}^{*}\PssRtwoDRtwo d_{1}H(0,0)}{d_{1}d_{2}-\alpha_{1}\widehat{\alpha}_{2}\alpha_{2}\widehat{\alpha}_{1}}.
\end{array}
\end{displaymath}
The following Theorem summarizes the main result of this subsection.
\begin{theorem}\label{rhbv}
Under stability conditions given in Theorem \ref{Thm2users}, $H(x,0)$ is given in terms of the solution of a homogeneous Riemann-Hilbert boundary value problem
\begin{equation}
\begin{array}{rl}
H(x,0)=&\frac{\lambda_{1}d_{2}+\alpha_{1}\widehat{\alpha}_{2}(\bar{t}_{1}\bar{t}_{2}\alpha_{2}^{*}\PssRtwoDRtwo-\lambda_{2})}{(\bar{x}-1)^{r}\bar{t}_{1}\bar{t}_{2}(\alpha_{1}\widehat{\alpha}_{2}\alpha_{2}\widehat{\alpha}_{1}-d_{1}d_{2})}\left((\bar{x}-x)^{r}\exp[\frac{\gamma(x)-\gamma(1)}{2i\pi}\int_{|t|=1}\frac{\log\{J(t)\}}{(t-\gamma(x))(t-\gamma(1))}dt]\right.\vspace{2mm}\\&\left.+\frac{\alpha_{1}^{*}\PssRoneDRone d_{2}\bar{x}^{r}}{\alpha_{1}\widehat{\alpha}_{2}\alpha_{2}^{*}\PssRtwoDRtwo}\exp[\frac{-\gamma(1)}{2i\pi}\int_{|t|=1}\frac{\log\{J(t)\}}{t(t-\gamma(1))}dt]\right),
\end{array}
\label{fin}
\end{equation}
where $\bar{x}$ is given in \ref{ap4}, and $\gamma(.)$ a conformal mapping, see \ref{conf}. A similar expression can be derived for $H(0,y)$ by solving another Riemann-Hilbert boundary value problem on the closed contour $\mathcal{L}$. Then, using (\ref{we}) we uniquely obtain $H(x,y)$.
\end{theorem}
\begin{proof}
See \ref{riemann}
\end{proof}

\subsection{Performance metrics}\label{per}
In the following we derive formulas for the expected number of packets and the average delay at each relay node in steady state, say $E_{i}$ and $D_{i}$, $i=1,2,$ respectively. Denote by $H_{1}(x,y)$, $H_{2}(x,y)$ the derivatives of $H(x,y)$ with respect to $x$ and $y$ respectively. Then, $E_{i}=H_{i}(1,1)$, and using Little's law $D_{i}=H_{i}(1,1)/\lambda_{i}$, $i=1,2$. Using the functional equation (\ref{we}), (\ref{r1}) and (\ref{r2}) we arrive after simple calculations in
\begin{equation}
\begin{array}{lccr}
E_{1}=\frac{\lambda_{1}+d_{1}H_{1}(1,0)}{\bar{t}_{1}\bar{t}_{2}\alpha_{1}\widehat{\alpha}_{2}},&&&E_{2}=\frac{\lambda_{2}+d_{2}H_{2}(0,1)}{\bar{t}_{1}\bar{t}_{2}\alpha_{1}\widehat{\alpha}_{2}}.
\end{array}
\label{perf}
\end{equation} 
We only focus on $E_{1}$, $D_{1}$ (similarly we can obtain $E_{2}$, $D_{2}$). Note that $H_{1}(1,0)$ can be obtained using either (\ref{fin}) or (\ref{soll}). For the general case $\frac{\alpha_{1}\widehat{\alpha}_{2}}{\alpha_{1}^{*}\PssRoneDRone}+\frac{\alpha_{2}\widehat{\alpha}_{1}}{\alpha_{2}^{*}\PssRtwoDRtwo}\neq1$, and using (\ref{fin}),
\begin{equation}
\begin{array}{rl}
H_{1}(1,0)=&\frac{\lambda_{1}d_{2}+\alpha_{1}\widehat{\alpha}_{2}(\bar{t}_{1}\bar{t}_{2}\alpha_{2}^{*}\PssRtwoDRtwo-\lambda_{2})}{\bar{t}_{1}\bar{t}_{2}(\alpha_{1}\widehat{\alpha}_{2}\alpha_{2}\widehat{\alpha}_{1}-d_{1}d_{2})}\frac{\gamma^{\prime}(1)}{2\pi i}\int_{|t|=1}\frac{\log\{J(t)\}}{(t-\gamma(1))^{2}}dt.
\end{array}
\label{xz}
\end{equation}
Substituting (\ref{xz}) in (\ref{perf}) we obtain $E_{1}$, and dividing with $\lambda_{1}$, the average delay $D_{1}$. Note that the calculation of (\ref{perf}) requires the evaluation of integrals (\ref{zx}), (\ref{cv}), (\ref{xz}) using the trapezoid rule, and $\gamma(1)$, $\gamma^{\prime}(1)$, as described above. 

\section{Explicit expressions for the symmetrical model}\label{sym}
In the following we consider the symmetrical model and obtain exact expressions for the average delay without solving a boundary value problem.

As symmetrical, we mean the case where  $\widehat{\lambda}_{k}=\widehat{\lambda}$, $\PfkDk=1-\PskDk=1-q=\bar{q}$, $\PskDonetwo=\tilde{q}$, $t_{k}=t$, $k=1,2,$ $\alpha_{i}^{*}=\alpha^{*}$, $\alpha_{i}=\alpha$, $P_{s}(R_{i},D,\{R_{i}\})=\bar{s}$, $P_{s}(R_{i},D,\{R_{1},R_{2}\})=s_{1,2}$, $P_{s}^{*}(R_{i},D,\{R_{i}\})=\tilde{s}$, $P_{s}(k,R_{i},\{k\})=P_{s}(1,R_{i},\{1,2\})+P_{s}(2,R_{i},\{1,2\})=P_{s,k}(k,R_{i},\{1,2\})=r$, $k=1,2$, $i=1,2$. As a result, $d_{1}=d_{2}=d$. 

Then, by exploiting the symmetry of the model we clearly have $H_{1}(1,1)=H_{2}(1,1)$, $H_{1}(1,0)=H_{2}(0,1)$. Recall that $E_{i}=H_{i}(1,1)$ the expected number of packets in relay node $R_{i}$. Therefore, after simple calculations using (\ref{we}) we obtain,
\begin{equation}
\begin{array}{c}
E_{1}=\frac{\widehat{\lambda}+t(t+2\bar{t}\bar{q})r+\bar{t}^{2}dH_{1}(1,0)}{\bar{t}^{2}\alpha
\widehat{\alpha}-(\widehat{\lambda}+t(t+2\bar{t}\bar{q})r)}.
\end{array}
\label{t1}
\end{equation}
Setting $x=y$ in (\ref{we}), differentiating it with respect to $x$ at $x=1$, and using (\ref{r1}) we obtain
\begin{equation}
\begin{array}{c}
E_{1}+E_{2}=2E_{1}=\frac{2(\widehat{\lambda}+t(t+2\bar{t}\bar{q})-\widehat{\lambda}^{2}+2H_{1}(1,0)\bar{t}^{2}(\alpha\widehat{\alpha}+d)}{2[\bar{t}^{2}\alpha
\widehat{\alpha}-(\widehat{\lambda}+t(t+2\bar{t}\bar{q})r]}.
\end{array}
\label{t2}
\end{equation}
Using (\ref{t1}), (\ref{t2}) we finally obtain
\begin{equation}
\begin{array}{c}
E_{1}=E_{2}=\frac{\widehat{\lambda}^{2}d+2\widehat{\lambda}\alpha\widehat{\alpha}+\lambda(\lambda+2\bar{\lambda}\bar{q})r(2\alpha\widehat{\alpha}-rd)}{2\alpha^{*}\tilde{s}[\bar{\lambda}^{2}\alpha
\widehat{\alpha}-(\widehat{\lambda}+\lambda(\lambda+2\bar{\lambda}\bar{q})r)]}.
\end{array}
\label{rt}
\end{equation}
Therefore, using Little's law the average delay in a node is given by
\begin{equation}
\begin{array}{c}
D_{1}=D_{2}=\frac{\widehat{\lambda}^{2}d+2\widehat{\lambda}\alpha\widehat{\alpha}+t(t+2\bar{t}\bar{q})r(2\alpha\widehat{\alpha}-rd)}{2\tilde{\lambda}\alpha^{*}\tilde{s}[\bar{t}^{2}\alpha
\widehat{\alpha}-\lambda]},
\end{array}
\label{rt1}
\end{equation}
where $\lambda=\widehat{\lambda}+t(t+2\bar{t}\bar{q})r$, and $\bar{t}^{2}\alpha
\widehat{\alpha}-\lambda>0$ due to the stability condition.
\section{Focusing on idle relay nodes}\label{idle}
An important aspect in the management of a cooperative network is to investigate the idle slots of relays (i.e., when at least one of the relays is empty). This is crucial for a variety of reasons related to load balancing as well as to energy conservation of relays. In the following, we show how to obtain information about the stationary distribution of the hitting point process of $\mathbf{N}$, which is a Markov chain formed by the successive hitting points of the boundary $W=W_{0,1}\cup W_{0,0}\cup W_{1,0}$, where $W_{0,1}=\{(N_{1},N_{2}):N_{1}=0,N_{2}>0\}$, $W_{0,0}=\{(N_{1},N_{2}):N_{1}=0,N_{2}=0\}$, $W_{1,0}=\{(N_{1},N_{2}):N_{1}>0,N_{2}=0\}$. 

Denote by $t_{m}$ the hitting epochs of $\mathbf{N}_{n}$ with its boundary $W$, and $\{\mathbf{k}_{m},m=1,2,...\}$, $\mathbf{k}_{m}=(k_{m}^{(1)},k_{m}^{(2)})=\mathbf{N}_{t_{m}}$, with $t_{1}=0$, $\mathbf{k}_{1}\in W$. Then, the sequence $\mathbf{k}_{m}$ is the hitting point process \cite{cohe} of $\mathbf{N}_{n}$, and possesses a stationary distribution when $\mathbf{N}_{n}$, and possesses a stationary distribution. Let $\mathbf{s}=(s_{1},s_{2})$ a stochastic vector with distribution the stationary distribution of $\{\mathbf{k}_{m}\}$. Then, for $|z|\leq 1$, let $Q_{1}(z)=E(z^{s_{1}}1_{\{\mathbf{s}\in W_{1,0}\}})$, $Q_{2}(z)=E(z^{s_{2}}1_{\{\mathbf{s}\in W_{0,1}\}})$, $Q_{0}=E(1_{\{\mathbf{s}\in W_{0,0}\}})$. Then, $P((N_{1},N_{2})\in W)=1-E(1_{\{N_{1}>0,N_{2}>0\}})$. Simple arguments imply that,
\begin{displaymath}
\begin{array}{l}
Q_{1}(z)=\frac{H(z,0)-H(0,0)}{1-E(1_{\{N_{1}>0,N_{2}>0\}})},\,Q_{2}(z)=\frac{H(0,z)-H(0,0)}{1-E(1_{\{N_{1}>0,N_{2}>0\}})},\,Q_{0}=\frac{H(0,0)}{1-E(1_{\{N_{1}>0,N_{2}>0\}})}.
\end{array}
\end{displaymath}
Moreover, equation (\ref{we}) can be rewritten as
\begin{equation}
\begin{array}{c}
(xy-\Psi(x,y))[H(x,y)-H(x,0)-H(0,y)+H(0,0)]=xy[\tilde{C}(xy)-1]H(0,0)\\
+[y\tilde{A}(x,y)-xy][H(x,0)-H(0,0)]+[x\tilde{B}(x,y)-xy][H(0,y)-H(0,0)],
\end{array}
\end{equation}
or equivalently,
\begin{equation}
\begin{array}{c}
(xy-\Psi(x,y))\frac{E(x^{N_{1}}y^{N_{2}}1_{\{N_{1}>0,N_{2}>0\}})}{H(0,0)}=xy[\tilde{C}(xy)-1]\\
+[y\tilde{A}(x,y)-xy]\frac{Q_{1}(x)}{H(0,0)}+[x\tilde{B}(x,y)-xy]\frac{Q_{2}(y)}{H(0,0)},
\end{array}
\label{szx}
\end{equation}
with $Q_{0}+Q_{1}(1)+Q_{2}(1)=1$, and where,
\begin{displaymath}
\begin{array}{rl}
\Psi(x,y)=&Z(x,y)[xy(1+(1-x)L_{1}+(1-y)L_{2}+(1-xy)L_{3})\\&+\bar{t}_{1}\bar{t}_{2}[\alpha_{1}\widehat{\alpha}_{2}y(1-x)+\alpha_{2}\widehat{\alpha}_{1}x(1-y)],\\
\widehat{A}(x,y)=&x(1+(1-x)L_{1}+(1-y)L_{2}+(1-xy)L_{3})+\bar{t}_{1}\bar{t}_{2}\alpha^{*}_{1}(1-x)P_{s}^{*}(R_{1},D,\{R_{1}\}),\\
\widehat{B}(x,y)=&y(1+(1-x)L_{1}+(1-y)L_{2}+(1-xy)L_{3})+\bar{t}_{1}\bar{t}_{2}\alpha^{*}_{2}(1-y)P_{s}^{*}(R_{2},D,\{R_{2}\}),\\
\widehat{C}(x,y)=&1+(1-x)L_{1}+(1-y)L_{2}+(1-xy)L_{3}.
\end{array}
\end{displaymath}
Equation (\ref{szx}) can be solved following a similar approach by reducing this problem in a Riemann boundary value problem; see also \cite{coh}. Recall that $Q_{1}(x)$ (resp. $Q_{2}(y)$) is the generating function of the stationary distribution of the hitting points of $W_{1,0}$ (resp. $W_{0,1}$). 
\section{Numerical results}\label{num}
In this section we evaluate numerically the theoretical results obtained in the previous sections. We focus on a ``symmetric" setup in order to simplify the presentation. In particular, we assume that $\alpha^*=1$, $\alpha=0.7$, and $t=0.1$. Table \ref{tab1} summarizes the basic setup.

\begin{table}
\caption{Topology Setup}
\label{tab1}
	\begin{tabular}{|r|l|}
		\hline
		\bfseries Parameters - Description & \bfseries Value \\
		\hline
		Distance between the sources and the destination &$110m$\\
		\hline Distance between the sources and the relays&  $80m$ \\
	\hline	Distance between relays and destination& $80m$\\
		\hline Path-loss exponent&$4$\\
\hline Transmit power of the relays & $10mW$ \\
\hline Transmit power of the sources & $1mW$ \\
		\hline
	\end{tabular}	
\end{table}
\begin{table}
\caption{Success probabilities}
\label{tab2}
	\begin{tabular}{|r|c|c|}
		\hline
		\bfseries  & \bfseries $SINR_t=0.2$& $SINR_t=1$ \\
		\hline
		$P_s(D,1)$&$0.74$&0.23\\
		\hline $P_s(R_1,1)=P_s(R_2,1)$&  $0.92$&$0.66$ \\
		\hline	$P_{s}(R_{i},D,\{R_{i}\})$& $0.99$&$0.96$\\
		\hline $P_{s}(R_{i},D,\{R_{1},R_{2}\})$&$0.83$&$0.5$\\
		\hline
	\end{tabular}	
\end{table}
 
We also consider two cases for the $SINR_{t}$ threshold, say $0.2$ and $1$\footnote{Note that when $SINR_t=0.2$ the MPR capability is stronger thus, we can have more than two concurrent successful transmissions.}.
Using equation (1) in \cite{PappasTWC2015} the success probabilities are obtained for each value of SINR. Table \ref{tab2} contains the values of the success probabilities for each $SINR_t$ setup.

\subsection{Stability}
Here we present the effect of the number of sources on the stability region. We consider the cases where the number of sources is varying from $1$ to $11$, i.e. $N=1,...,11$. In Fig. \ref{stab02}, we consider the case where $SINR_t=0.2$, the outer curve in the plot represents the case where $N=1$, the inner the case corresponds to $N=11$.
Since we have stronger MPR capabilities at the receivers we observe that the region for up to four sources is convex thus, it the performance is better than a time division scheme as also described in Sections \ref{stab-anal-2} and \ref{stab-anal-n}.
\begin{figure}[!ht]
\centering
\includegraphics[scale=0.4]{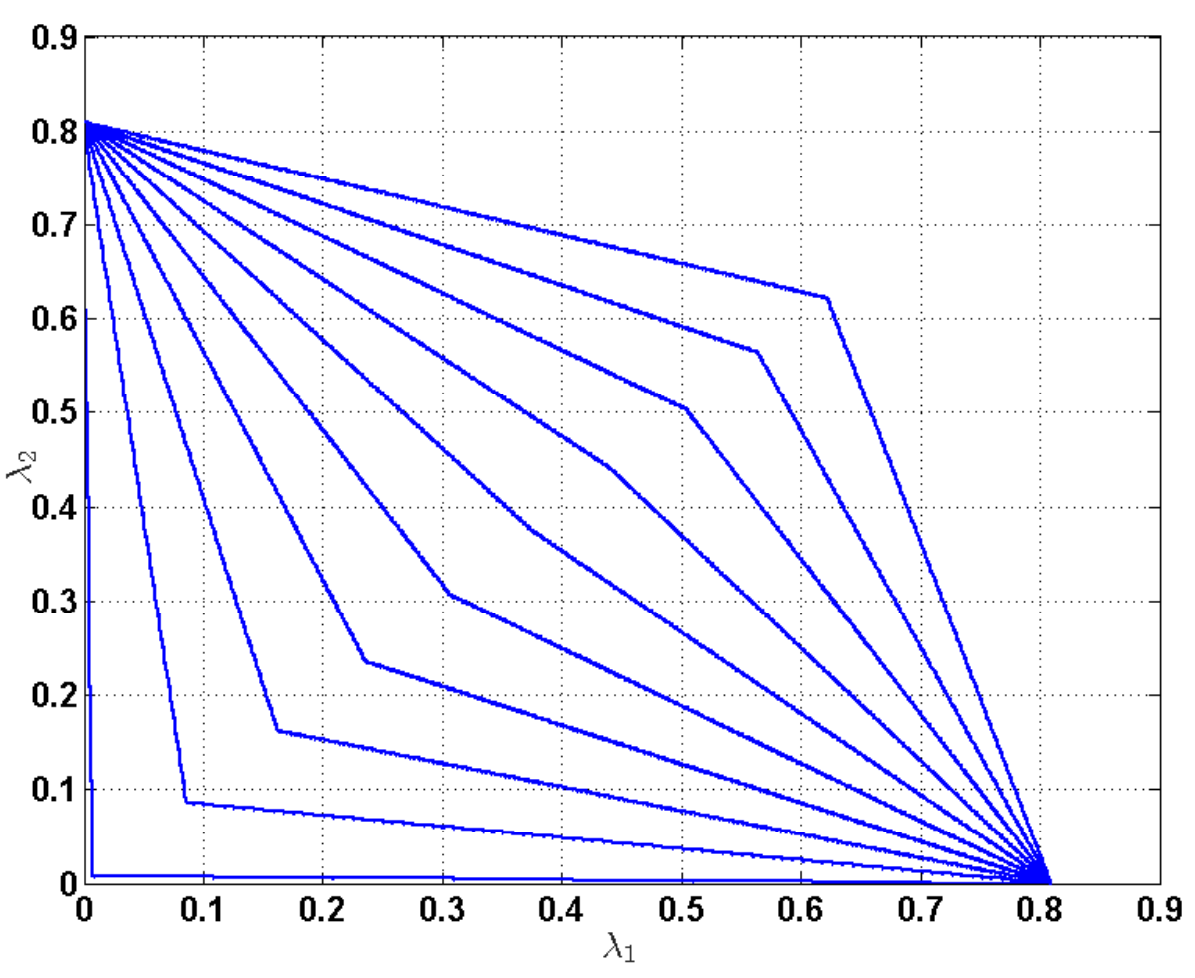}
\caption{The effect of number of sources on the stability region for $SINR_t=0.2$.}
\label{stab02}
\end{figure}
In Fig. \ref{stab1}, we consider the case where $SINR_t=1$, the outer curve in the plot represents the case where $N=1$ and the inner the case for $N=11$. In this plot, we observe that for more than two sources the region is not a convex set. Thus, a time division scheme would be preferable as also described in Sections \ref{stab-anal-2} and \ref{stab-anal-n}.
\begin{figure}[!ht]
\centering
\includegraphics[scale=0.4]{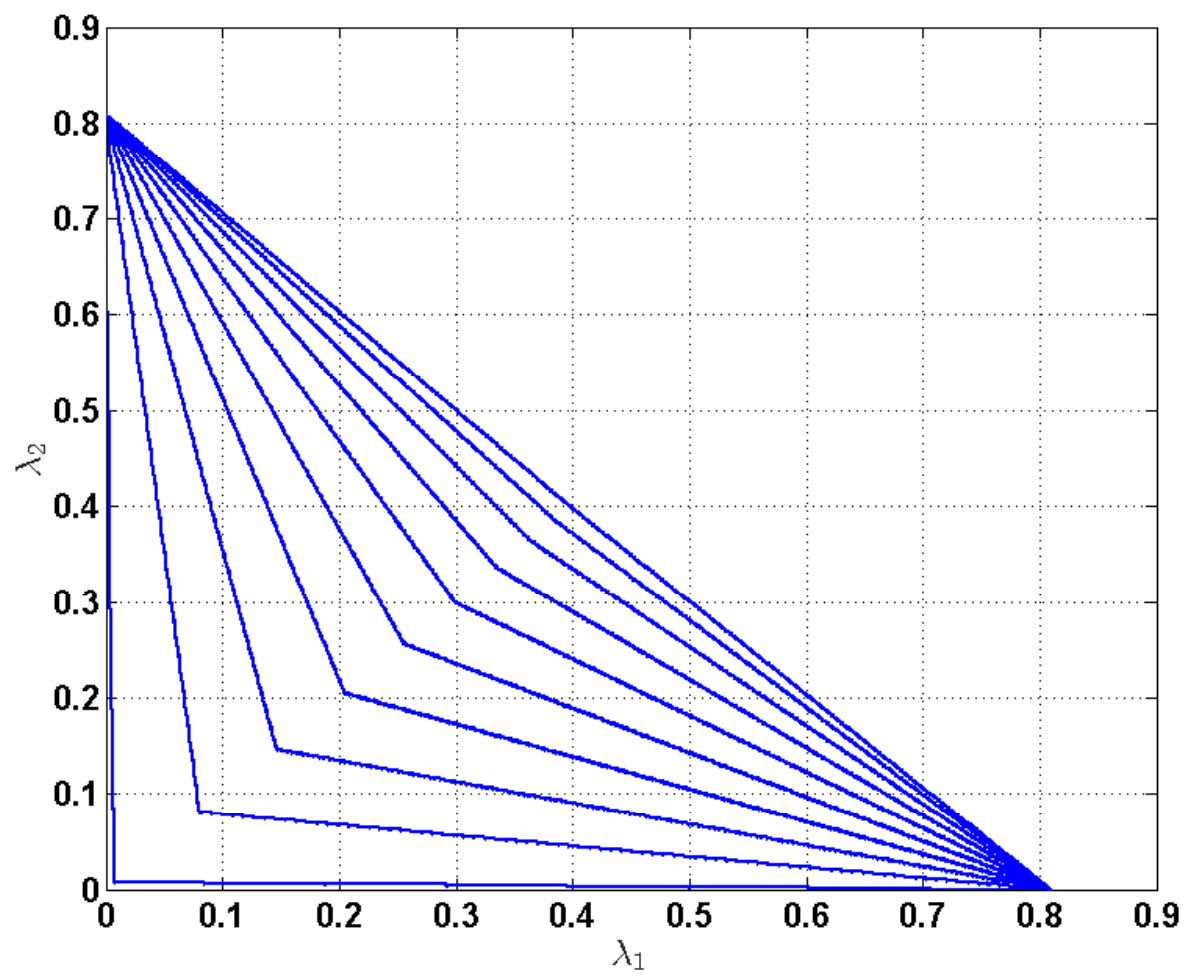}
\caption{The effect of number of sources on the stability region for $SINR_t=1$.}
\label{stab1}
\end{figure}

In both cases, we observe that as the number of sources is increasing, then the number of slots that the relays can transmit packets from their queues is decreasing. Thus, when $N=11$, the stability region is approaching the empty set, which is an indication that the relays cannot sustain the traffic in their queues.
\subsection{Throughput performance}
We provide numerical evaluation regarding throughput per source and aggregate throughput for the case of pure relays, i.e. $\widehat{\lambda}_{1}=\widehat{\lambda}_{2}=0$.

The throughput per source, as a function of the number of sources in the network is depicted in Fig. \ref{thrU}. We observe that the throughput per source is decreasing as the number of sources is increasing. Moreover, we also observe that for $SINR_t=0.2$, the system becomes unstable after $12$ sources, while for $SINR_t=1$ the system remains stable when the number of sources is up to $6$. The aggregate throughput is depicted in Fig. \ref{athr}. Note that the maximum aggregate throughput for $SINR_t=0.2$ and $SINR_t=1$ is achieved for twelve and six sources respectively.
\begin{figure}[!ht]
\centering
\includegraphics[scale=0.4]{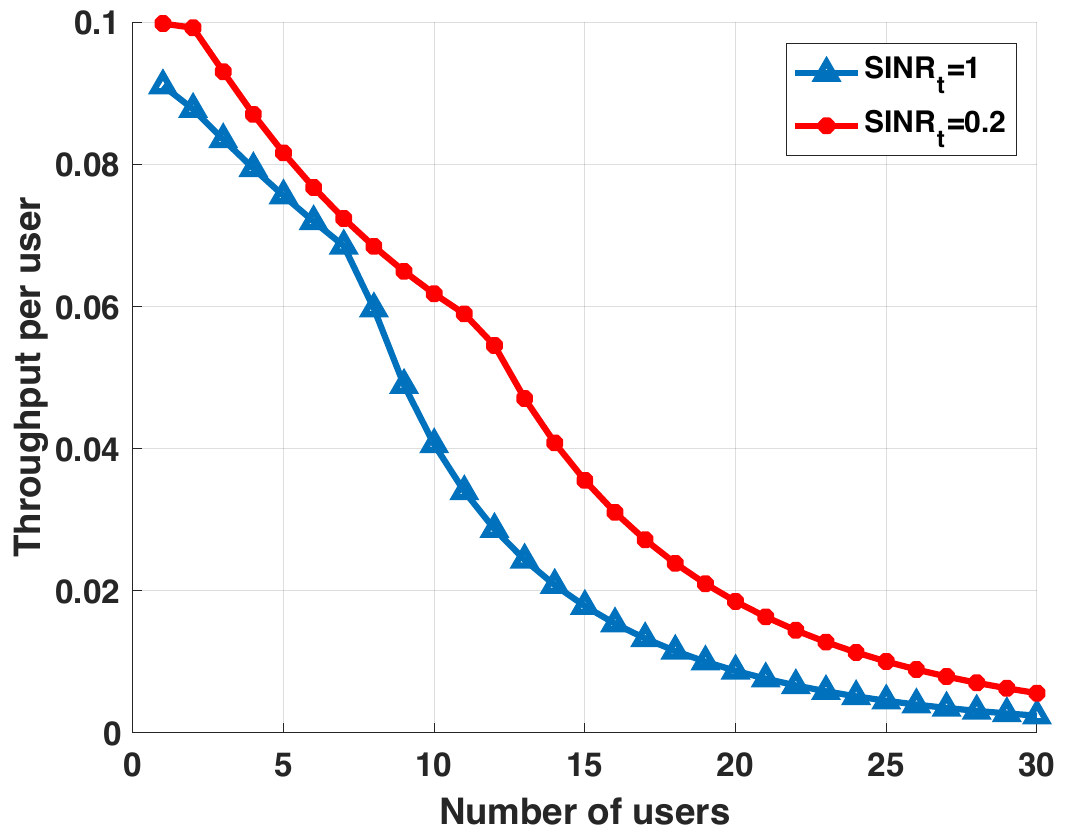}
\caption{Throughput per source versus the number of sources.}
\label{thrU}
\end{figure}

\begin{figure}[!ht]
\centering
\includegraphics[scale=0.4]{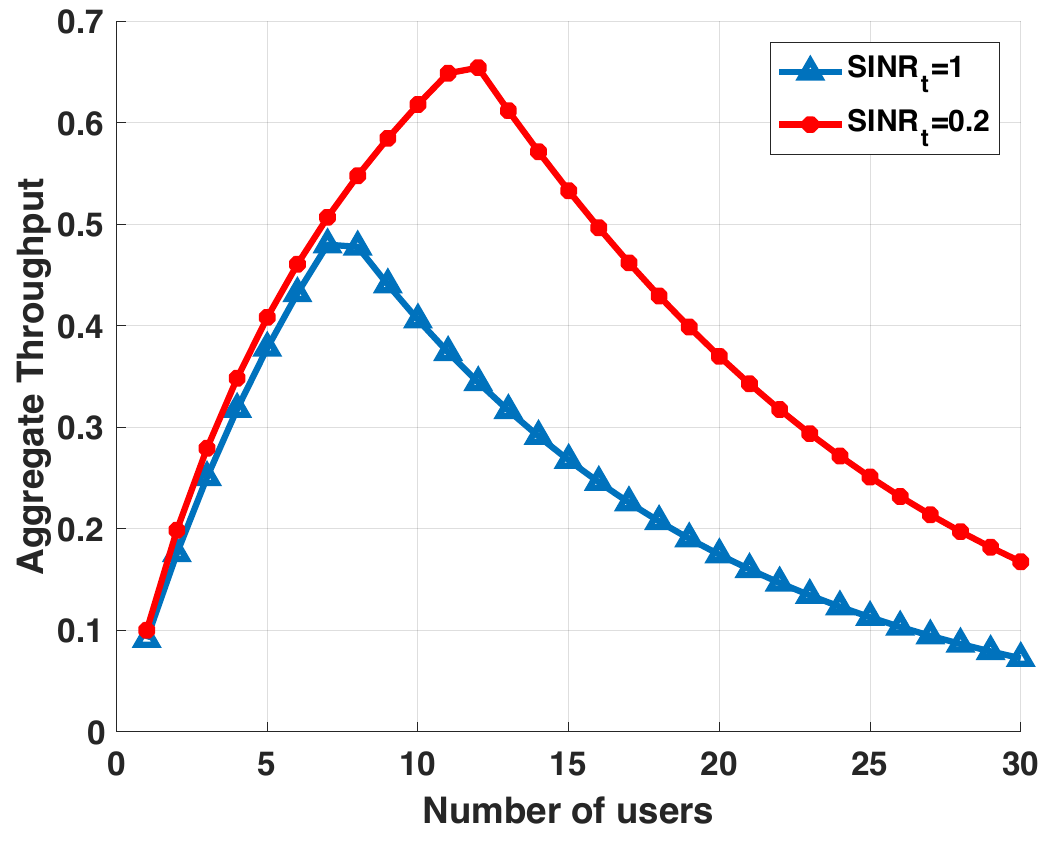}
\caption{Aggregate throughput versus the number of sources.}
\label{athr}
\end{figure}
\subsection{Average Delay and Stability Region for the capture model} 
In this part we will evaluate the average delay performance. The setup will be different from the previous two subsection due to the capture channel model assumed in the analysis.
\paragraph{Example 1. The symmetrical system} In the following we consider the symmetric system and we investigate the effect of the system parameters on the average delay. We assume that $q=0.5$, $\bar{s}=0.8$, $\tilde{s}=0.9$, $s_{12}=0.4$. In Fig. \ref{f153} we can see the effect of $r$ (i.e., the reception probability of blocked packet by a relay node) on the average delay for increasing values of $\widehat{\lambda}$ (i.e., the average number of of external packet arrivals at a relay node during a time slot) and $\alpha$ (i.e., the transmission probability of a relay). As expected, the increase in $\widehat{\lambda}$ increases the expected delay, and that decrease becomes more apparent as $\alpha$ takes small values and $r$ increases.
\begin{figure}[!ht]
\centering
\includegraphics[scale=0.43]{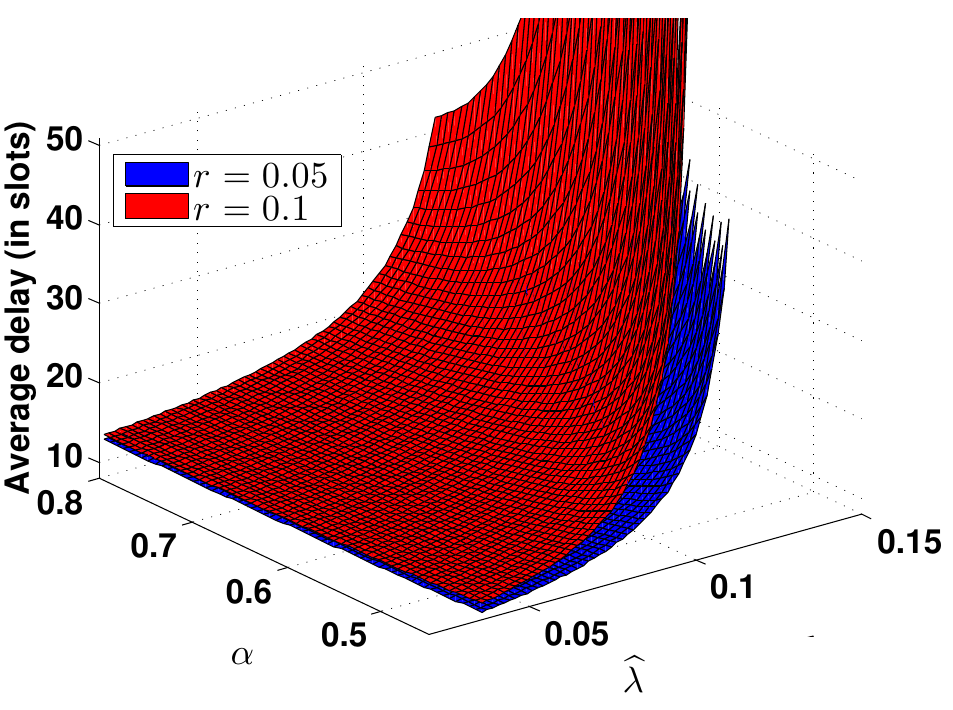}
\caption{The average delay vs $\alpha$ and $\widehat{\lambda}$ for different values of $r$.}
\label{f153}
\end{figure}

Figure \ref{f152} illustrate how sensitive is the average delay as we increase the probability of a direct transmission (at the beginning of a slot) of a source. In particular, as $t$ remains small, the increase in $\widehat{\lambda}$ from $0.1$ to $0.15$ will not affect the average delay. As $t$ increases, the simultaneous increase in $\widehat{\lambda}$ will cause a rapid increase in the average delay, even when we set the transmission probability $\alpha=\alpha^{*}$. This is expected, since at the beginning of a slot, sources have precedence over the relays.
\begin{figure}[!ht]
\centering
\includegraphics[scale=0.4]{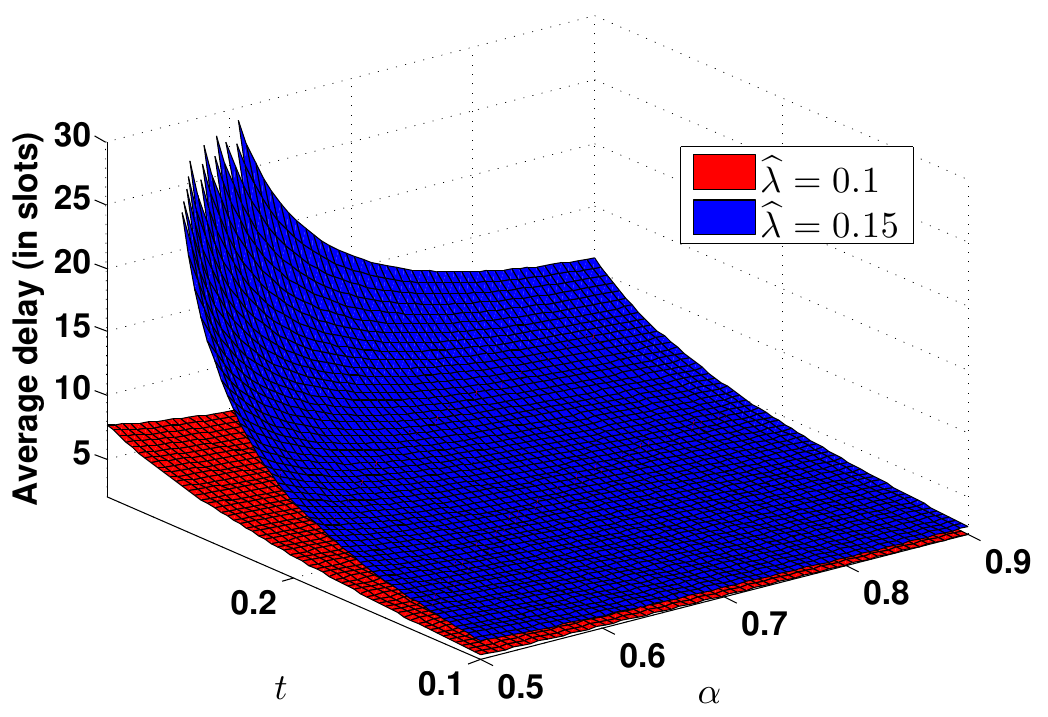}
\caption{Effect of $\widehat{\lambda}$ on average delay.}
\label{f152}
\end{figure}

Similar observations can be deduced by Fig. \ref{f151}, where we can observe the average delay as a function of $\alpha^{*}$ and $\widehat{\lambda}$. Clearly, as $t$ increases from $0.3$ to $0.4$, the average delay increases rapidly, especially when, $\widehat{\lambda}$ tends to $0.1$.
\begin{figure}[!ht]
\centering
\includegraphics[scale=0.43]{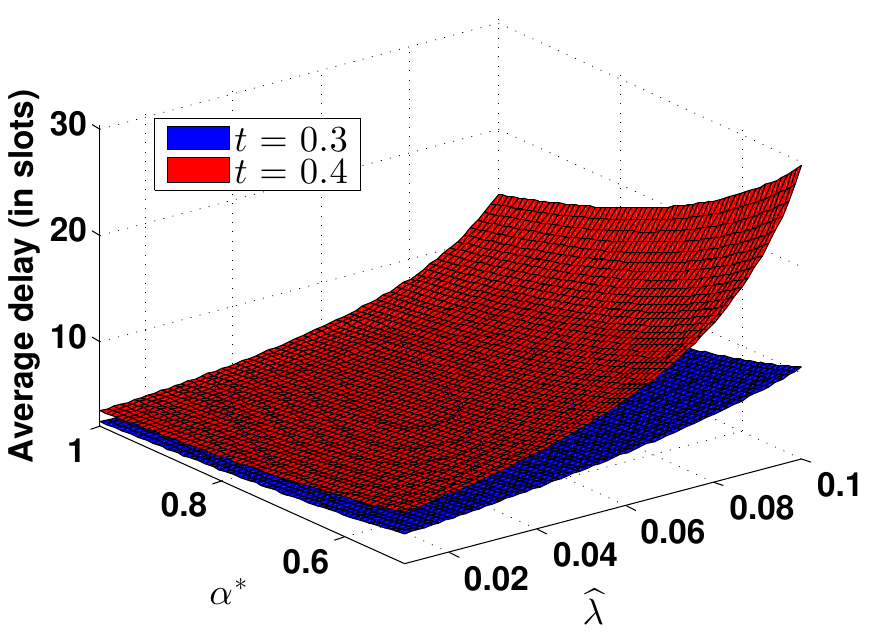}
\caption{Effect of $t$ on average delay.}
\label{f151}
\end{figure}
\paragraph{Example 2. Stability region}
We now focus on the general model, and specifically on the case $\frac{\alpha_{1}\widehat{\alpha}_{2}}{\alpha_{1}^{*}\tilde{s}_{1/\{1\}}}+\frac{\alpha_{2}\widehat{\alpha}_{1}}{\alpha_{2}^{*}\tilde{s}_{2/\{2\}}}\neq1$. We investigate the effect of parameters on the stability region, i.e., the set of arrival rate vectors $(\lambda_{1},\lambda_{2})$, for which the queues of the system are stable. In what follows, let $\alpha_{1}=0.7$, $\alpha_{2}=0.6$, $\alpha_{2}^{*}=0.9$, $\PssRoneDRone=\PssRtwoDRtwo=0.9$, $\PsRoneDRone=\PsRtwoDRtwo=0.8$, $P_{s}(R_{i},D,\{R_{1},R_{2}\})=0.4$, $i=1,2$, $t_{2}=0.3$.

In Fig. \ref{f15s} we observe the impact of $t_1$, i.e., the packet transmission probability of source $P_1$ at the beginning of a slot, on the stability region. It is easily seen that by changing this factor, we heavily affect the network performance. Indeed, although the destination node hears both sources and the relays, it gives priority to the sources at the beginning of a time slot, and thus, the increase of $t_1$ from $0.2$ to $0.4$ causes a deterioration of the stability region. Moreover, that increase will affect both relays, i.e., both average arrival rates at the relays are decreasing in order to sustain stability.

\begin{figure}[!ht]
\centering
\includegraphics[scale=0.43]{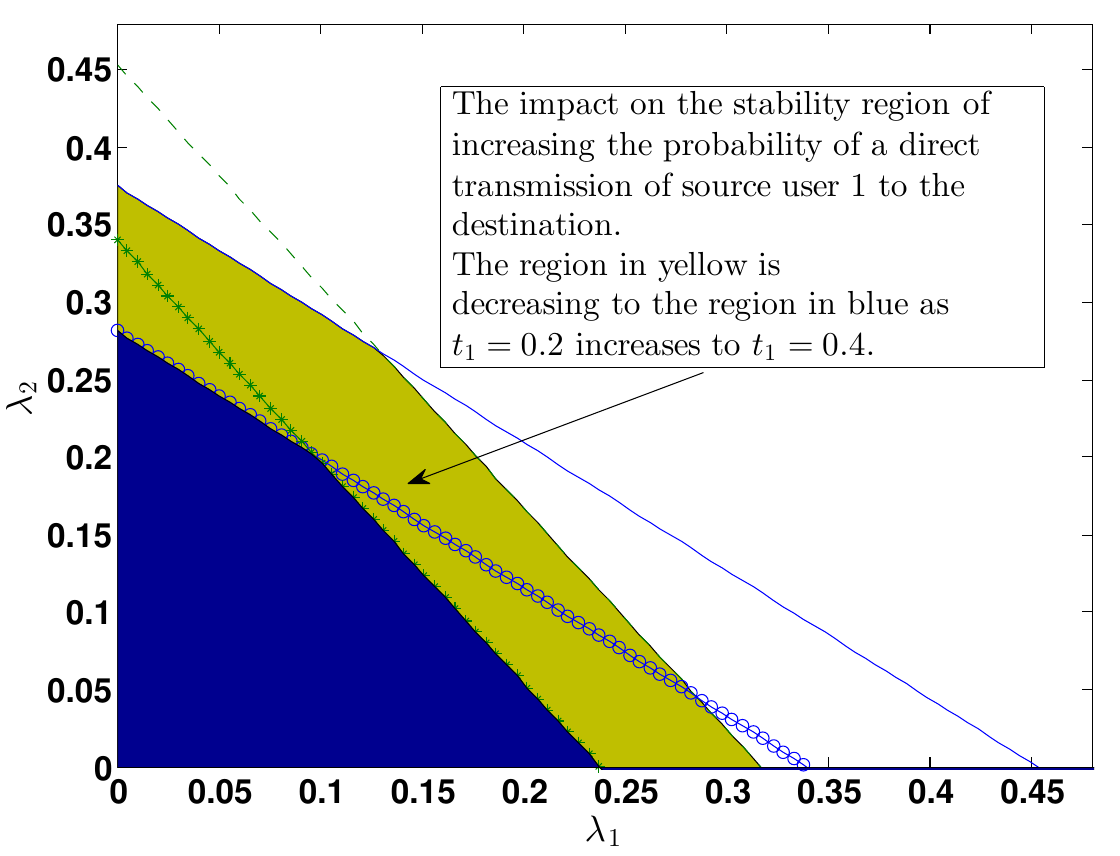}
\caption{Effect of $t_{1}$ on the stability region for $\alpha_{1}^{*}=0.9$.}
\label{f15s}
\end{figure}
\begin{figure}[!ht]
\centering
\includegraphics[scale=0.4]{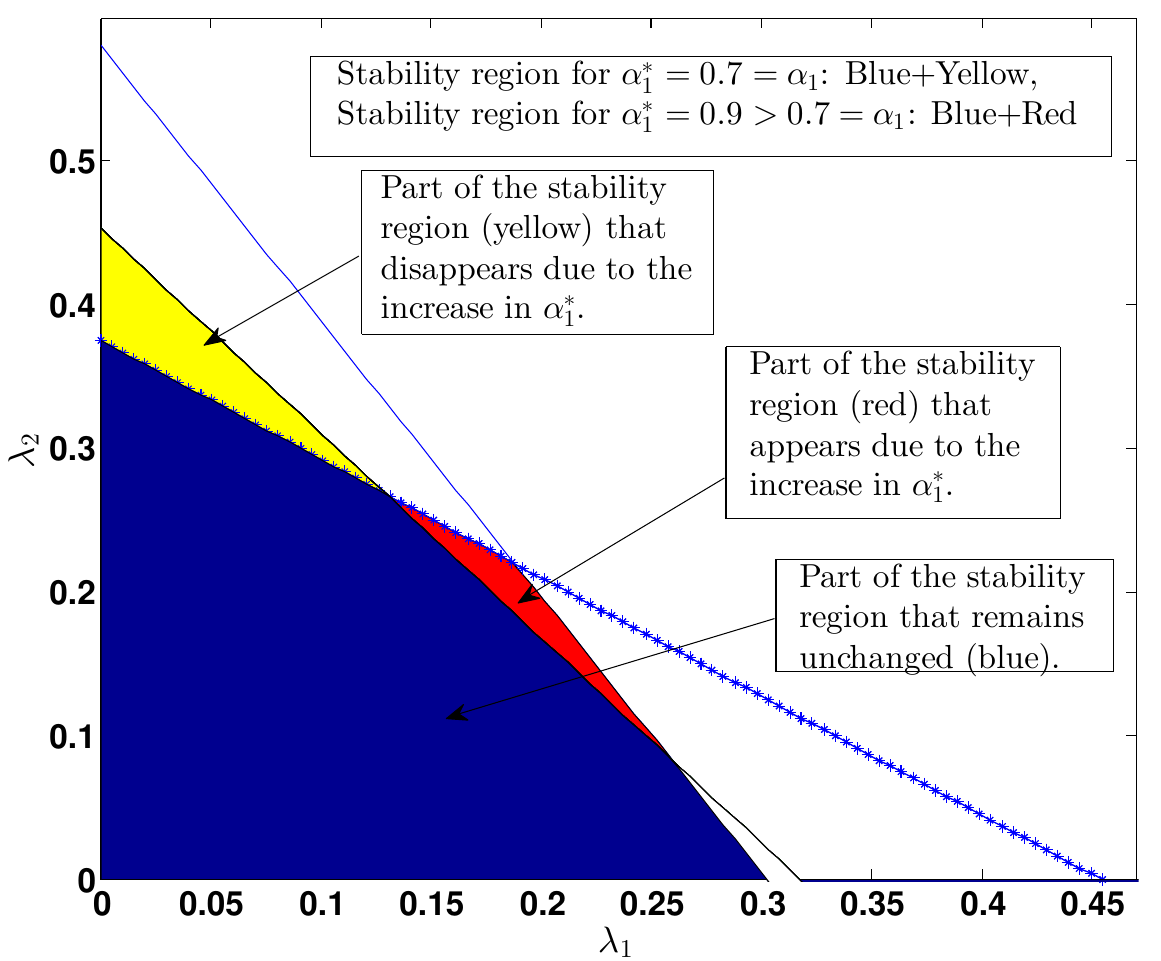}
\caption{Effect of transmission control on the stability region.}
\label{f15z}
\end{figure}

Finally, Fig. \ref{f15z} demonstrates the impact of the adaptive transmission control on the stability region when we set $t_{1}=0.2$. More precisely, we compare the stability regions obtained for the network with adaptive transmission policy, with the one obtained for the network with non-adaptive transmission policy for $R_{1}$. In the latter case, we assume that $\alpha_{1}^{*}=\alpha_{1}=0.7$, i.e., the relay $R_1$ does not adapt its transmission probability when it senses relay $R_2$ inactive. In such case, the stability region is the part in Fig. \ref{f15z} colored in blue and yellow. In the former case, when relay $R_1$ adapts its transmission probability to $\alpha_{1}^{*}=0.9$ (i.e., when it senses relay $R_2$ inactive) the stability region is given by the part of Fig. \ref{f15z} colored in blue and red. Note that the increase in $\alpha_{1}^{*}$ will affect the performance of relay $R_2$, since a packet in relay $R_1$ is more likely to be transmitted in a slot. Thus, we expect that the arrival rate for the relay $R_{2}$ to be lower in order to ensure stability.   

\section{Conclusions and future work}

In this work, we focused on the performance analysis of a relay-assisted cooperative wireless network with MPR capabilities and adaptive transmission policy. By applying the stochastic dominance technique we obtained the stability region under general MPR both for the asymmetric network of two-sources, two-relay nodes, and for the symmetric model of $N$-sources. In addition, we provided the aggregate throughput and the throughput per source in terms of closed form expressions.

We investigated the fundamental problem of characterizing the delay performance in an asymmetric network of two-sources, two-relays, by performing a detailed mathematical analysis, which led to the determination of the generating function of the stationary joint queue length distribution of relay nodes with the aid of the theory of boundary value problems. For the symmetrical network, closed form expressions for the expected delay at each relay node were derived without the need of solving a boundary value problem. We obtained extensive numerical results providing insights in the system performance.

In a future work we plan to generalize our work to the case of a completely random access network with no coordination among sources and relays. A challenging task is the investigation of the delay for a network with arbitrary number of relay nodes. This work will serve as a building block in order to obtain bounds for the expected delay at relay nodes. Moreover, we aim to investigate the time the relays get empty for the first time. This result will help us to optimize their transmission parameters in order to minimize the delay.

\section*{Acknowledgment}
This work has been partially supported by the European Union’s Horizon 2020 research and innovation programme and by the European Union’s Horizon 2020 research and innovation programme under the Marie Skłodowska-Curie grant agreement No. 643002 (ACT5G). This work was supported in part by ELLIIT and CENIIT.

\appendix
\section{Proof of Theorem \ref{Thm2users}}\label{proofThm2users}
The average service rates of the first and second relay are given in (\ref{eq:mu1}). Since the average service rate of each relay depends on the queue size of the other relay, the stability region cannot be computed directly. Thus, we apply the stochastic dominance technique introduced in \cite{Rao_TIT1988}, i.e. we construct hypothetical dominant systems, in which the relay with the empty queue transmits dummy packets, while the non-empty relay transmits according to its traffic.

In the first dominant system, the first relay transmit dummy packets and the second relay behaves as in the original system. All the rest operational aspects remain unaltered in the dominant system. Thus, in this dominant system, the first queue never empties, hence the service rate for the second relay is
\begin{displaymath}
\begin{array}{rl}
\mu_2 = &\overline{t}_1 \overline{t}_2 \alpha_{2} \left[  \alpha_{1} \PsRtwoDRoneRtwo + \overline{\alpha}_{1} \PsRtwoDRtwo \right]\\
=&\overline{t}_1 \overline{t}_2 \alpha_{2} \left[  \alpha_{1} \Delta_2+  \PsRtwoDRtwo \right].
\end{array}
\end{displaymath}

Then, we can obtain stability conditions for the second relay by applying Loynes' criterion \cite{Loynes}. The queue at the second source is stable if and only if $\lambda_2 < \mu_2$, that is $\lambda_2 < \overline{t}_1 \overline{t}_2 \alpha_{2} \left[  \alpha_{1} \Delta_2+ \PsRtwoDRtwo \right]$. Then, we can obtain the probability that the second relay is empty by applying Little's law, i.e.
\begin{equation} \label{eq:Pr2empty_D1}
\begin{array}{c}
\mathrm{Pr}\left(N_2 = 0 \right)  = 1-\frac{\lambda_2}{\overline{t}_1 \overline{t}_2 \alpha_{2} \left[  \alpha_{1} \Delta_2+  \PsRtwoDRtwo \right]}.
\end{array}
\end{equation}

After replacing \eqref{eq:Pr2empty_D1} into \eqref{eq:mu1} we obtain 
\begin{displaymath}
\begin{array}{c}
\mu_1 = \overline{t}_1 \overline{t}_2 \alpha_{1}^{*} \PssRoneDRone - \frac{\lambda_{2}\alpha_{1}^*\PssRoneDRone}{\Bone} + \frac{\lambda_{2}\alpha_{1} \left[\Btwo\right]}{\Bone}.
\end{array}
\end{displaymath}

Thus, after applying Loynes' criterion, the stability condition for the first relay in the first dominant system is 
\begin{displaymath}
\begin{array}{c}
\lambda_{1} < \overline{t}_1 \overline{t}_2 \alpha_{1}^{*} \PssRoneDRone - \frac{\lambda_{2} \left[\alpha_{1}^*\PssRoneDRone - \lambda_{2}\alpha_{1} \left[\Btwo\right]\right]}{\Bone}.
\end{array}
\end{displaymath}
Therefore, the stability region $\mathcal{R}_1$ obtained from the first dominant system is given in (\ref{eq:R1_2})

Similarly, we construct a second dominant system where the second relay transmits a dummy packet when it is empty and the first relay behaves as in the original system. All other operational aspects remain unaltered in the dominant system. Following the same steps as in the first dominant system, the stability region, $\mathcal{R}_2$, of the second dominant system will be given in (\ref{eq:R1_2}).

An important observation made in \cite{Rao_TIT1988} is that the stability conditions obtained by the stochastic dominance technique are not only sufficient but also necessary for the stability of the original system. 
The \emph{indistinguishability} argument~\cite{Rao_TIT1988} applies to our problem as well. Based on the construction of the dominant system, \emph{it is easy to see that the queue sizes in the dominant system are always greater than those in the original system, provided they are both initialized to the same value and the arrivals are identical in both systems}. Therefore, given $\lambda_{2}<\mu_{2}$, if for some $\lambda_{1}$, the queue at the first relay is stable in the dominant system, then the corresponding queue in the original system must be stable. Conversely, if for some $\lambda_{1}$ in the dominant system, the queue at the first relay saturates, then it will not transmit dummy packets, and as long as the first relay has a packet to transmit, the behavior of the dominant system is identical to that of the original system since dummy packet transmissions are eliminated as we approach the stability boundary. Therefore, the original and the dominant systems are indistinguishable at the boundary points.
\section{Derivation of the Functional equation}\label{po}
The queue evolution equation (\ref{x}) implies
\begin{equation}
\begin{array}{l}
E(x^{N_{1,n+1}}y^{N_{2,n+1}})=D(x,y)\left\{\bar{t}_{1}\bar{t}_{2}[P(N_{1,n}=N_{2,n}=0)+E(x^{N_{1,n}}1_{\{N_{1,n}>0,N_{2,n}=0\}})[1\right.\vspace{2mm}\\\left.+\alpha_{1}^{*}\PssRoneDRone(\frac{1}{x}-1)]
+E(y^{N_{2,n}}1_{\{N_{1,n}=0,N_{2,n}>0\}})(1+\alpha_{2}^{*}\PssRtwoDRtwo(\frac{1}{y}-1))\right.\vspace{2mm}\\\left.+E(x^{N_{1,n}}y^{N_{2,n}}1_{\{N_{1,n}>0,N_{2,n}>0\}})(1+\alpha_{1}\widehat{\alpha}_{2}(\frac{1}{x}-1)+\alpha_{2}\widehat{\alpha}_{1}(\frac{1}{y}-1))]\right.\vspace{2mm}\\
\left.+E(x^{N_{1,n}}y^{N_{2,n}})[t_{1}\bar{t}_{2}(1+S_{1}(x,y))+t_{2}\bar{t}_{1}(1+S_{2}(x,y))+t_{2}t_{1}(1+S_{3}(x,y))]\right\},
\end{array}
\label{fc}
\end{equation}
where $1_{\{A\}}$ denotes the indicator function of the event $A$, and
\begin{displaymath}
\begin{array}{rl}
S_{k}(x,y)=&\PfkDk[\PskRonek\overline{P}_{s}(k,R_{2},\{k\})(x-1)+\overline{P}_{s}(k,R_{1},\{k\}) P_{s}(k,R_{2},\{k\})(y-1)\\
&+P_{s}(k,R_{1},\{k\}) P_{s}(k,R_{2},\{k\})(xy-1)],\,k=1,2,\vspace{2mm}\\
S_{3}(x,y)=&(x-1)\{\PfDonetwo\PfRtwoonetwo(\PsoneRoneonetwo+\PstwoRoneonetwo)\\
&+\PsoneDonetwo\overline{P}_{s,2}(2,R_{2},\{1,2\})P_{s,2}(2,R_{1},\{1,2\})\\&+\PstwoDonetwo\overline{P}_{s,1}(1,R_{2},\{1,2\})P_{s,1}(2,R_{1},\{1,2\})\}\\
&+(y-1)\{\PfDonetwo\PfRoneonetwo(\PsoneRtwoonetwo+\PstwoRtwoonetwo)\\
&+\PsoneDonetwo\overline{P}_{s,2}(2,R_{1},\{1,2\})P_{s,2}(2,R_{2},\{1,2\})\\&+\PstwoDonetwo\overline{P}_{s,1}(1,R_{1},\{1,2\})P_{s,1}(2,R_{2},\{1,2\})\}\\
&+(xy-1)\{\PfDonetwo(\PsoneRoneonetwo+\PstwoRoneonetwo)(\PsoneRtwoonetwo\\&+\PstwoRtwoonetwo)+\PsoneDonetwo P_{s,2}(2,R_{2},\{1,2\})P_{s,2}(2,R_{1},\{1,2\})\\&+\PstwoDonetwo P_{s,1}(1,R_{2},\{1,2\})P_{s,1}(2,R_{1},\{1,2\})\}.
\end{array}
\end{displaymath}
Note that 
\begin{displaymath}
\begin{array}{rl}
H(0,0)=&\lim_{n\to\infty}P(N_{1,n}=N_{2,n}=0),\\
H(x,0)-H(0,0)=&\lim_{n\to\infty}E(x^{N_{1,n}}1_{\{N_{1,n}>0,N_{2,n}=0\}}),\\
H(0,y)-H(0,0)=&\lim_{n\to\infty}E(y^{N_{2,n}}1_{\{N_{1,n}=0,N_{2,n}>0\}}),\\
H(x,y)-H(x,0)-H(0,y)+H(0,0)=&\lim_{n\to\infty}E(x^{N_{1,n}}y^{N_{2,n}}1_{\{N_{1,n}>0,N_{2,n}>0\}}).
\end{array}
\end{displaymath}

Then, using (\ref{fc}) we obtain the functional equation (\ref{we}).

\section{Proof of Lemma \ref{LEM}}\label{a0}
It is easily seen that $R(x,y)=\frac{xy-\Psi(x,y)}{xyD(x,y)}$, where $\Psi(x,y)=D(x,y)[xy(1+L_{3}(xy-1))+y(1-x)(\bar{t}_{1}\bar{t}_{2}\alpha_{1}\widehat{\alpha}_{2}-L_{1}x)+x(1-y)(\bar{t}_{1}\bar{t}_{2}\alpha_{2}\widehat{\alpha}_{1}-L_{2}y)]$, where for $|x|\leq1$, $|y|\leq1$, $\Psi(x,y)$ is a generating function of a proper probability distribution. Now, for $|y|=1$, $y\neq1$ and $|x|=1$ it is clear that $|\Psi(x,y)|<1=|xy|$. Thus, from Rouch\'e's theorem, $xy-\Psi(x,y)$ has exactly one zero inside the unit circle. Therefore, $R(x,y)=0$ has exactly one root $x=X_{0}(y)$, such that $|x|<1$. For $y=1$, $R(x,1)=0$ implies
\begin{displaymath}
\begin{array}{c}
(x-1)\left(\lambda_{1}-\frac{\bar{t}_{1}\bar{t}_{2}\alpha_{1}\widehat{\alpha}_{2}}{x} \right)=0.
\end{array}
\end{displaymath}
Therefore, for $y=1$, and since $\lambda_{1}<\bar{t}_{1}\bar{t}_{2}\alpha_{1}\widehat{\alpha}_{2}$, the only root of $R(x,1)=0$ for $|x|\leq1$, is $x=1$.\hfill$\square$
\section{Proof of Lemma \ref{SQ}}\label{a1}
We will prove the part related to $\mathcal{M}$. Similarly, we can also prove the other. For $y\in[y_{1},y_{2}]$, $D_{y}(y)$ is negative, so $X_{0}(y)$, $X_{1}(y)$ are complex conjugates. Therefore, $|X(y)|^{2}=\frac{\widehat{c}(y)}{\widehat{a}(y)}=g(y)$. Clearly, $g(y)$ is an increasing function for $y\in[0,1]$ and thus, $|X(y)|^{2}\leq g(y_{2})=\beta_{0}$. Using simple algebraic considerations we can prove that, $X_{0}(y_{1}):=\beta_{1}=-g(y_{1})$ is the extreme left point of $\mathcal{M}$. Finally, $\zeta(\delta)$ is derived by solving $Re(X(y))=-\widehat{b}(y)/2\widehat{a}(y)$ for $y$ with $\delta = Re(X(y))$, and taking the solution such that $y\in[0,1]$.\hfill$\square$
\section{Intersection points of the curves}\label{ap4}
In the following, we focus on the location of the intersection points of $R(x,y)=0$, $A(x,y)=0$ (resp. $B(x,y)$). These points (if they exist) are potential singularities for the functions $H(x,0)$, $H(0,y)$, and thus, their investigation is crucial regarding the analytic continuation of $H(x,0)$, $H(0,y)$ outside the unit disk; see also Lemma 2.2.1 and Theorem 3.2.3 in \cite{fay1} for alternative approaches. 
\subsection{Intersection points between $R(x,y)=0$, $A(x,y)=0$.}
Let $R(x,y) = 0$, $x = X_{\pm}(y)$, $y\in \tilde{C}_{y}$. We can easily show that the resultant in $x$ of the two polynomials $R(x,y)$ and $A(x,y)$ is $Res_{x}(R,A;y)=y(y-1)Q(y)$, where
\begin{displaymath}
\begin{array}{rl}
Q(y)=&-d_{1}[\lambda_{2}d_{1}+a_{2}\widehat{a}_{1}(\widehat{\lambda}_{2}(1+\widehat{\lambda}_{1})+L_{2})]y^{2}+ya_{2}\widehat{a}_{1}[d_{1}(\lambda_{1}+\widehat{\lambda}_{2}(1+\widehat{\lambda}_{1})+L_{2})\\&-\bar{t}_{1}\bar{t}_{2}a_{1}^{*}P_{s}^{*}(R_{1},D,\{R_{1}\})(a_{2}\widehat{a}_{1}+d_{1})]+\bar{t}_{1}\bar{t}_{2}a_{1}^{*}P_{s}^{*}(R_{1},D,\{R_{1}\})(a_{2}\widehat{a}_{1})^{2}.
\end{array}
\end{displaymath}
Note that $Q(0)=\bar{t}_{1}\bar{t}_{2}a_{1}^{*}P_{s}^{*}(R_{1},D,\{R_{1}\})(a_{2}\widehat{a}_{1})^{2}>0$ and
\begin{displaymath}
Q(1)=d_{1}[\lambda_{1}\alpha_{2}\widehat{\alpha}_{1}-\lambda_{2}d_{1}-\bar{t}_{1}\bar{t}_{2}\alpha_{2}\widehat{\alpha}_{1}a_{1}^{*}P_{s}^{*}(R_{1},D,\{R_{1}\})]>0,
\end{displaymath}
since $d_{1}<0$ and due to the stability condition. Similarly, for $R(x,y) = 0$, $y = Y_{\pm}(x)$, $x\in \tilde{C}_{x}$, the resultant in $y$ of the two polynomials $R(x,y)$, $A(x,y)$ is $Res_{y}(R,A;x)=x(x-1)\bar{t}_{1}\bar{t}_{2}\alpha_{2}\widehat{\alpha}_{1}\tilde{Q}(x)$, where,
\begin{displaymath}
\begin{array}{rl}
\tilde{Q}(x)=&-[\alpha_{2}\widehat{\alpha}_{1}\lambda_{1}+(\widehat{\lambda}_{1}(1+\widehat{\lambda}_{2})+L_{1})d_{1}]x^{2}+x[(\widehat{\lambda}_{1}(1+\widehat{\lambda}_{2})+\lambda_{2}+L_{1})d_{1}\\&+(\alpha_{2}\widehat{\alpha}_{1}+d_{1})\alpha_{1}^{*}P_{s}^{*}(R_{1},D,\{R_{1}\})\bar{t}_{1}\bar{t}_{2}]-\alpha_{1}^{*}P_{s}^{*}(R_{1},D,\{R_{1}\})d_{1}\bar{t}_{1}\bar{t}_{2}.
\end{array}
\end{displaymath}

Note also that $\tilde{Q}(0)=-\alpha_{1}^{*}P_{s}^{*}(R_{1},D,\{R_{1}\})d_{1}\bar{t}_{1}\bar{t}_{2}>0$ since $d_{1}<0$ and $$\tilde{Q}(1)=\bar{t}_{1}\bar{t}_{2}\alpha_{2}\widehat{\alpha}_{1}\alpha_{1}^{*}P_{s}^{*}(R_{1},D,\{R_{1}\})-\lambda_{1}\alpha_{2}\widehat{\alpha}_{1}+\lambda_{2}d_{1}>0,$$ due to the stability conditions (see Lemma \ref{Thm2users}). If $\alpha_{1}^{*}\leq min\{1,\frac{\alpha_{2}\widehat{\alpha}_{1}\lambda_{1}+(\widehat{\lambda}_{1}(1+\widehat{\lambda}_{2})+L_{1})\alpha_{1}\widehat{\alpha}_{2}}{(\widehat{\lambda}_{1}(1+\widehat{\lambda}_{2})+L_{1})P_{s}^{*}(R_{1},D,\{R_{1}\})}\}$, then $\lim_{x\to\infty}\tilde{Q}(x)=-\infty$, and $\tilde{Q}(x)=0$ has two roots of opposite sign, say $x_{*}<0<1<x^{*}$. If $\frac{\alpha_{2}\widehat{\alpha}_{1}\lambda_{1}+(\widehat{\lambda}_{1}(1+\widehat{\lambda}_{2})+L_{1})\alpha_{1}\widehat{\alpha}_{2}}{(\widehat{\lambda}_{1}(1+\widehat{\lambda}_{2})+L_{1})P_{s}^{*}(R_{1},D,\{R_{1}\})}<\alpha_{1}^{*}\leq 1$, then $\lim_{x\to\infty}\tilde{Q}(x)=+\infty$, and $\tilde{Q}(x)=0$ has two positive roots, say $1<\tilde{x}_{*}<x_{3}<x_{4}<\tilde{x}^{*}$, due to the stability conditions. In the former case we have to check if $x^{*}$ is in $S_{x}$, while in the latter case if $\tilde{x}_{*}$ is in $S_{x}$. These zeros, if they lie in $S_{x}$ such that $|Y_{0}(x)|\leq1$, are poles of $H(x,y)$. Denote by 
\begin{displaymath}
\bar{x}=\left\{\begin{array}{rl}
x^{*},&\alpha_{1}^{*}\leq min\{1,\frac{\alpha_{2}\widehat{\alpha}_{1}\lambda_{1}+(\widehat{\lambda}_{1}(1+\widehat{\lambda}_{2})+L_{1})\alpha_{1}\widehat{\alpha}_{2}}{(\widehat{\lambda}_{1}(1+\widehat{\lambda}_{2})+L_{1})P_{s}^{*}(R_{1},D,\{R_{1}\})}\},\\
\tilde{x}_{*},&\frac{\alpha_{2}\widehat{\alpha}_{1}\lambda_{1}+(\widehat{\lambda}_{1}(1+\widehat{\lambda}_{2})+L_{1})\alpha_{1}\widehat{\alpha}_{2}}{(\widehat{\lambda}_{1}(1+\widehat{\lambda}_{2})+L_{1})P_{s}^{*}(R_{1},D,\{R_{1}\})}<\alpha_{1}^{*}\leq 1.
\end{array}\right.
\end{displaymath}
\subsection{Intersection points between $R(x,y)=0$, $B(x,y)=0$.} Let $y\in \tilde{C}_{y}$ and $R(x,y) = 0$, $x = X_{\pm}(y)$. It is easily shown that the resultant in $x$ of $R(x,y)$, $B(x,y)$ is $Res_{x}(R,B,y)=y(y-1)T(y)$, where
\begin{displaymath}
\begin{array}{rl}
T(y)=&-[\alpha_{1}\widehat{\alpha}_{2}\lambda_{2}+(\widehat{\lambda}_{2}(1+\widehat{\lambda}_{1})+L_{2})d_{2}]y^{2}+y[(\widehat{\lambda}_{2}(1+\widehat{\lambda}_{1})+\lambda_{1}+L_{2})d_{2}\\&+(\alpha_{1}\widehat{\alpha}_{2}+d_{2})\alpha_{2}^{*}P_{s}^{*}(R_{2},D,\{R_{2}\})\bar{t}_{1}\bar{t}_{2}]-\alpha_{2}^{*}P_{s}^{*}(R_{2},D,\{R_{2}\})d_{2}\bar{t}_{1}\bar{t}_{2}.
\end{array}
\end{displaymath}

Note that $T(0)=-\alpha_{2}^{*}P_{s}^{*}(R_{2},D,\{R_{2}\})d_{2}\bar{t}_{1}\bar{t}_{2}>0$, $T(1)=\bar{t}_{1}\bar{t}_{2}\alpha_{1}\widehat{\alpha}_{2}\alpha_{2}^{*}P_{s}^{*}(R_{2},D,\{R_{2}\})-\lambda_{2}\alpha_{1}\widehat{\alpha}_{2}+\lambda_{1}d_{2}>0$, since $d_{2}<0$ and due to the stability conditions. If $\alpha_{2}^{*}<min\{1,\frac{\alpha_{1}\widehat{\alpha}_{2}\lambda_{2}+(\widehat{\lambda}_{2}(1+\widehat{\lambda}_{1})+L_{2})\alpha_{2}\widehat{\alpha}_{1}}{(\widehat{\lambda}_{2}(1+\widehat{\lambda}_{1})+L_{2})P_{s}^{*}(R_{2},D,\{R_{2}\})}\}$, $\lim_{y\to \infty}T(x)=-\infty$, and $T(x)$ has two roots of opposite sign, say $y_{*}$, $y^{*}$ such that $y_{*}<0<1<y^{*}$, which in turn implies that $B(X_{0}(y),y)\neq0$, $y\in[y_{1},y_{2}]\subset(0,1)$, or equivalently $B(x,Y_{0}(x))\neq0$, $x\in\mathcal{M}$. When $\frac{\alpha_{1}\widehat{\alpha}_{2}\lambda_{2}+(\widehat{\lambda}_{2}(1+\widehat{\lambda}_{1})+L_{2})\alpha_{2}\widehat{\alpha}_{1}}{(\widehat{\lambda}_{2}(1+\widehat{\lambda}_{1})+L_{2})P_{s}^{*}(R_{2},D,\{R_{2}\})}<\alpha_{2}^{*}\leq1$, $\lim_{y\to \infty}T(y)=+\infty$, and $T(y)$ has two positive roots, say $\widehat{y}_{*}$, $\widehat{y}^{*}$ such that $1<\widehat{y}_{*}<y_{3}<y_{4}<\widehat{y}^{*}$, which in turn implies that $B(X_{0}(y),y)\neq0$, $y\in[y_{1},y_{2}]$, i.e., $B(x,Y_{0}(x))\neq0$, $x\in\mathcal{M}$.\hfill$\square$
\section{Proof of Theorem \ref{dirbv}}\label{diri}
For $y\in \mathcal{D}_{y}=\{x\in\mathcal{C}:|y|\leq1,|X_{0}(y)|\leq1\}$,
\begin{equation}
\begin{array}{l}
\alpha_{2}\widehat{\alpha}_{1}H(X_{0}(y),0)+d_{2}H(0,y)+\frac{\alpha_{2}\widehat{\alpha}_{1}(1-\rho)C(X_{0}(y),y)}{A(X_{0}(y),y)}=0.
\end{array}
\label{con}
\end{equation}
It is easily realised that for $y\in \mathcal{D}_{y}-[y_{1},y_{2}]$ both $H(X_{0}(y),0)$, and $H(0,y)$ are analytic and thus, by means of analytic continuation, we can also consider (\ref{con}) for $y\in[y_{1},y_{2}]$, or equivalently, for $x\in\mathcal{M}$
\begin{equation}
\begin{array}{c}
\alpha_{2}\widehat{\alpha}_{1}H(x,0)+d_{2}H(0,Y_{0}(x))+\frac{\alpha_{2}\widehat{\alpha}_{1}(1-\rho)C(x,Y_{0}(x))}{A(x,Y_{0}(x))}=0.
\end{array}
\label{con2}
\end{equation}
Then, multiplying both sides of (\ref{con2}) by the imaginary complex number $i$, and noticing that $H(0,Y_{0}(x))$ is real for $x\in\mathcal{M}$, since $Y_{0}(x)\in[y_{1},y_{2}]$, we have
\begin{equation}
\begin{array}{c}
Re(iH(x,0))=Re(-i\frac{C(x,Y_{0}(x))}{A(x,Y_{0}(x))})(1-\rho),\,x\in\mathcal{M}.
\end{array}
\label{p1}
\end{equation}
To proceed, we have to check for poles of $H(x,0)$ in $S:=G_{\mathcal{M}}\cap\bar{D}_{x}^{c}$, and $D_{x}=\{x:|x|<1\}$, $\bar{D}_{x}=\{x:|x|\leq1\}$, $\bar{D}_{x}^{c}=\{x:|x|>1\}$. These poles, if exist, they coincide with the zeros of $A(x,Y_{0}(x))$ in $S_{x}$; see \ref{ap4}.
Note that equation (\ref{p1}) is defined on $\mathcal{M}$. In order to solve (\ref{p1}) we must firstly conformally transform the problem from $\mathcal{M}$ to the unit circle. Let the conformal mapping, $z=\gamma(x):G_{\mathcal{M}}\to G_{\mathcal{C}}$, and its inverse $x=\gamma_{0}(z):G_{\mathcal{C}}\to G_{\mathcal{M}}$\footnote{For more details on the construction of the conformal mappings see \ref{conf}}. 

Then, we have the following problem: Find a function $\tilde{T}(z)=H(\gamma_{0}(z),0)$ regular for $z\in G_\mathcal{C}$, and continuous for $z\in\mathcal{C}\cup G_\mathcal{C}$ such that, $Re(i\tilde{T}(z))=w(\gamma_{0}(z))$, $z\in\mathcal{C}$. In case $H(x,0)$ has no poles in $S$, the solution of the Dirichlet problem with boundary condition (\ref{p1}) is:
\begin{equation}
\begin{array}{c}
H(x,0)=-\frac{1-\rho}{2\pi}\int_{|t|=1}f(t)\frac{t+\gamma(x)}{t-\gamma(x)}\frac{dt}{t}+C,\,x\in\mathcal{M},
\end{array}
\label{sol1}
\end{equation}
where $f(t)=Re(-i\frac{C(\gamma_{0}(t),Y_{0}(\gamma_{0}(t)))}{A(\gamma_{0}(t),Y_{0}(\gamma_{0}(t)))})$, $C$ a constant to be defined by setting $x=0\in G_{\mathcal{M}}$ in (\ref{sol1}) and using the fact that $H(0,0)=1-\rho$, $\gamma(0)=0$\footnote{In case $H(x,0)$ has a pole, say $\bar{x}$, we still have a Dirichlet problem for the function $(x-\bar{x})H(x,0)$.}. 

Following the discussion above,
\begin{displaymath}
\begin{array}{c}
C=(1-\rho)(1+\frac{1}{2\pi}\int_{|t|=1}f(t)\frac{dt}{t}),
\end{array}
\end{displaymath}
Setting $t=e^{i\phi}$, $\gamma_{0}(e^{i\phi})=\rho(\psi(\phi))e^{i\psi(\phi)}$, we obtain after some algebra,
\begin{displaymath}
\begin{array}{c}
f(e^{i\phi})=\frac{d_{1}\alpha_{2}^{*}\sin(\psi(\phi))(1-Y_{0}(\gamma_{0}(e^{i\phi}))^{-1})}{\rho(\psi(\phi))k(\phi)},\end{array}
\end{displaymath}
which is an odd function of $\phi$, and 
\begin{displaymath}
\begin{array}{c}
k(\phi)=[\alpha_{2}\widehat{\alpha}_{1}(1-Y_{0}^{-1}(\gamma_{0}(e^{i\phi})))+d_{1}(1-\frac{\cos(\psi(\phi))}{\rho(\psi(\phi))})]^{2}+(d_{1}\frac{\sin(\psi(\phi))}{\rho(\psi(\phi))})^{2}.
\end{array}
\end{displaymath}
Thus, $C=1-\rho$. Substituting in (\ref{sol1}), we obtain after simple calculations the integral representation of $H(x,0)$ on a real interval given in (\ref{sol1}).

Similarly, $H(0,y)$ is derived by solving another Dirichlet boundary value problem on the closed contour $\mathcal{L}$. Finally, using (\ref{we}) we uniquely obtain $H(x,y)$.
\section{Proof of Theorem \ref{rhbv}}\label{riemann}
Clearly,
\begin{equation}
\begin{array}{c}
A(X_{0}(y),y)G(X_{0}(y))=-B(X_{0}(y),y)L(y),\,y\in \mathcal{D}_{y}.
\end{array}
\label{po31}
\end{equation}
For $y\in \mathcal{D}_{y}-[y_{1},y_{2}]$ both $G(X_{0}(y))$, $L(y)$ are analytic and the right-hand side can be analytically continued up to the slit
$[y_1, y_2]$ or equivalently, for $x\in\mathcal{M}$
\begin{equation}
\begin{array}{c}
A(x,Y_{0}(x))G(x)=-B(x,Y_{0}(x))L(Y_{0}(x)).
\end{array}
\label{za1}
\end{equation}
Clearly, $G(x)$ is holomorphic for $D_{x}$, continuous for $\bar{D}_{x}$. However, $G(x)$ might has poles in $S_{x}$, based on the values of the system parameters. These poles (if exist) coincide with the zeros of $A(x,Y_{0}(x))$ in $S_{x}$; see \ref{ap4}. For $y\in[y_{1},y_{2}]$, let $X_{0}(y)=x\in\mathcal{M}$ and realize that $Y_{0}(X_{0}(y))=y$ so that $y=Y_{0}(x)$. Taking into account the possible poles of $G(x)$, and noticing that $L(Y_{0}(x))$ is real for $x\in\mathcal{M}$, since $Y_{0}(x)\in[y_{1},y_{2}]$, we have
\begin{equation}
\begin{array}{c}
Re[iU(x)\tilde{G}(x)]=0,\,x\in\mathcal{M},\vspace{2mm}\\
U(x)=\frac{A(x,Y_{0}(x))}{(x-\bar{x})^{r}B(x,Y_{0}(x))},\,\tilde{G}(x)=(x-\bar{x})^{r}G(x),
\end{array}
\label{df3}
\end{equation}
where $r=0,1$, whether $\bar{x}$ is zero or not of $A(x,Y_{0}(x))$ in $S_{x}$. Thus, $\tilde{G}(x)$ is regular for $x\in G_{\mathcal{M}}$, continuous for $x\in\mathcal{M}\cup G_{\mathcal{M}}$, and $U(x)$ is a non-vanishing function on $\mathcal{M}$. We then conformally transform the problem (\ref{df3}) from $\mathcal{M}$ to the unit circle, using the mapping $z=\gamma(x):G_{\mathcal{M}}\to G_{\mathcal{C}}$, and its inverse given by $x=\gamma_{0}(z):G_{\mathcal{C}}\to G_{\mathcal{M}}$\footnote{For more details on the construction of the conformal mappings see \ref{conf}}. 

Then, the problem in (\ref{df3}) is reduced to the following: Find a function $F(z):=\tilde{G}(\gamma_{0}(z))$, regular in $G_{\mathcal{C}}$, continuous in $G_{\mathcal{C}}\cup\mathcal{C}$ such that, $Re[iU(\gamma_{0}(z))F(z)]=0,\,z\in\mathcal{C}$. 

A crucial step in the solution of the boundary value problem is the determination of the index $\chi=\frac{-1}{\pi}[arg\{U(x)\}]_{x\in \mathcal{M}}$, where $[arg\{U(x)\}]_{x\in \mathcal{M}}$, denotes the variation of the argument of the function $U(x)$ as $x$ moves along the closed contour $\mathcal{M}$ in the positive direction, provided that $U(x)\neq0$, $x\in\mathcal{M}$. The value of the index is closely related to the stability conditions (see Theorem \ref{Thm2users}), and following the lines in \cite{fay} we have,
\begin{lemma}\begin{enumerate}
\item If $\lambda_{2}<\lambda_{2}^{*}$, then $\chi=0$ is equivalent to
\begin{displaymath}
\begin{array}{l}
\frac{d A(x,Y_{0}(x))}{dx}|_{x=1}<0\Leftrightarrow\lambda_{1}<\bar{t}_{1}\bar{t}_{2}[\alpha_{1}^{*}\PssRoneDRone+\frac{d_{1}\lambda_{2}}{\bar{t}_{1}\bar{t}_{2}\alpha_{2}\widehat{\alpha}_{1}}],\vspace{2mm}\\ \frac{d B(X_{0}(y),y)}{dy}|_{y=1}<0\Leftrightarrow\lambda_{2}<\bar{t}_{1}\bar{t}_{2}[\alpha_{2}^{*}\PssRtwoDRtwo+\frac{d_{2}\lambda_{1}}{\bar{t}_{1}\bar{t}_{2}\alpha_{1}\widehat{\alpha}_{2}}].
\end{array}
\end{displaymath}
\item If $\lambda_{2}\geq \lambda_{2}^{*}$, $\chi=0$ is equivalent to 
\begin{displaymath}
\begin{array}{c}
\frac{d B(X_{0}(y),y)}{dy}|_{y=1}<0\Leftrightarrow \lambda_{2}<\bar{t}_{1}\bar{t}_{2}[\alpha_{2}^{*}\PssRtwoDRtwo+\frac{d_{2}\lambda_{1}}{\bar{t}_{1}\bar{t}_{2}\alpha_{1}\widehat{\alpha}_{2}}].
\end{array}
\end{displaymath}
\end{enumerate}
\end{lemma}
Thus, under stability conditions (see Lemma \ref{Thm2users}) the problem defined in (\ref{df3}) has a unique solution for $x\in G_{\mathcal{M}}$ given by,
\begin{equation}
\begin{array}{rl}
H(x,0)=&K(x-\bar{x})^{-r}\exp[\frac{1}{2i\pi}\int_{|t|=1}\frac{\log\{J(t)\}}{t-\gamma(x)}dt]-\frac{\alpha_{1}^{*}\PssRoneDRone d_{2}H(0,0)}{d_{1}d_{2}-\alpha_{1}\widehat{\alpha}_{2}\alpha_{2}\widehat{\alpha}_{1}},
\end{array}
\label{sool1}
\end{equation}
where $K$ is a constant and $J(t)=\frac{\overline{U_{1}(t)}}{U_{1}(t)}$, $U_{1}(t)=U(\gamma_{0}(t))$, $|t|=1$. Setting $x=0$ in (\ref{sool1}) we derive a relation between $K$ and $H(0,0)$. Then, for $x=1\in G_{\mathcal{M}}$, and using the first in (\ref{rd}) we can obtain $K$ and $H(0,0)$. Substituting back in (\ref{sool1}) we obtain (\ref{fin}) for $x\in G_{\mathcal{M}}$. Similarly, we can determine $H(0,y)$ by solving another Riemann-Hilbert boundary value problem on the closed contour $\mathcal{L}$. Then, using (\ref{we}) we uniquely obtain $H(x,y)$.
\section{Construction of the conformal mappings}\label{conf}
The construction of the conformal mapping $\gamma(x)$ is not a trivial task. However, we can construct its inverse in order to obtain expressions for the expected value of the queue lengths in each relay node. To proceed, we need a representation of $\mathcal{M}$ in polar coordinates, i.e., $\mathcal{M}=\{x:x=\rho(\phi)\exp(i\phi),\phi\in[0,2\pi]\}.$ This procedure is described in detail in \cite{coh}. 

In the following we summarize the basic steps: Since $0\in G_{\mathcal{M}}$, for each $x\in\mathcal{M}$, a relation between its absolute value and its real part is given by $|x|^{2}=m(Re(x))$ (see Lemma \ref{SQ}). Given the angle $\phi$ of some point on $\mathcal{M}$, the real part of this point, say $\delta(\phi)$, is the zero of $\delta-\cos(\phi)\sqrt{m(\delta)}$, $\phi\in[0,2\pi].$ Since $\mathcal{M}$ is a smooth, egg-shaped contour, the solution is unique. Clearly, $\rho(\phi)=\frac{\delta(\phi)}{\cos(\phi)}$, and the parametrization of $\mathcal{M}$ in polar coordinates is fully specified.

 Then, the mapping from $z\in G_{\mathcal{C}}$ to $x\in G_{\mathcal{M}}$, where $z = e^{i\phi}$ and $x= \rho(\psi(\phi))e^{i\psi(\phi)}$, satisfying $\gamma_{0}(0)=0$, $\gamma_{0}(z)=\overline{\gamma_{0}(\bar{z})}$ is uniquely determined by (see \cite{coh}, Section I.4.4),
\begin{equation}
\begin{array}{rl}
\gamma_{0}(z)=&z\exp[\frac{1}{2\pi}\int_{0}^{2\pi}\log\{\rho(\psi(\omega))\}\frac{e^{i\omega}+z}{e^{i\omega}-z}d\omega],\,|z|<1,\\
\psi(\phi)=&\phi-\int_{0}^{2\pi}\log\{\rho(\psi(\omega))\}\cot(\frac{\omega-\phi}{2})d\omega,\,0\leq\phi\leq 2\pi,
\end{array}
\label{zx}
\end{equation}
i.e., $\psi(.)$ is uniquely determined as the solution of a Theodorsen integral equation with $\psi(\phi)=2\pi-\psi(2\pi-\phi)$. This integral equation has to be solved numerically by an iterative procedure. For the numerical evaluation of the integrals we split the interval $[0,2\pi]$ into $M$ parts of length $2\pi/M$, by taking $M$ points $\phi_{k}=\frac{2k\pi}{M}$, $k=0,1,...,M-1$. For the $M$ points given by their angles $\left\{\phi_{0},...,\phi_{M-1}\right\}$ we should solve the second in (\ref{zx}) to obtain the corresponding points $\left\{\psi(\phi_{0}),...,\psi(\phi_{M-1})\right\}$, iteratively from,
\begin{equation}
\begin{array}{rl}
\psi_{0}(\phi_{k})=&\phi_{k},\\
\psi_{n+1}(\phi_{k})=&\phi_{k}-\frac{1}{2\pi}\int_{0}^{2\pi}\log\left\{\frac{\delta(\psi_{n}(\omega))}{\cos(\psi_{n}(\omega))}\right\}\cot[\frac{1}{2}(\omega-\phi_{k})]d\omega,
\end{array}
\label{uip}
\end{equation}
where $\lim_{n\to\infty}\psi_{n+1}(\phi)=\psi(\phi)$, and $\delta(\psi_{n}(\omega))$ is determined by,
\begin{displaymath}
\delta(\psi_{n}(\omega))=cos(\psi_{n}(\omega))\sqrt{m(\delta(\psi_{n}(\omega)))},
\end{displaymath}
using the Newton-Raphson root finding method. For each step, the integral in (\ref{uip}) is
numerically determined by again using the trapezium rule with $M$ parts of equal length
$2\pi/M$. For the iteration, we have used the following stopping criterion $\max_{k\in\left\{0,1,...,M-1\right\}}\left|\psi_{n+1}(\phi_{k})-\psi_{n}(\phi_{k})\right|<10^{-6}$

Having obtained $\psi(\phi)$ numerically, the values of the conformal mapping $\gamma_{0}(z)$, $\left|z\right|\leq 1$, can be calculated by applying the Plemelj-Sokhotski formula to the first in (\ref{zx})
\begin{displaymath}
\gamma_{0}(e^{i\phi})=e^{i\psi(\phi)}\frac{\delta(\psi(\phi))}{\cos(\psi(\phi))}=\delta(\psi(\phi))[1+i \tan(\psi(\phi))],\,0\leq\phi\leq 2\pi.
\end{displaymath}
We further need to find $\gamma(1)$, $\gamma^{\prime}(1)$. To do this, one needs to use the Newton's method and solve $\gamma_{0}(z_{0})=1$, in $[0,1]$, i.e., $z_{0}$ is the zero in $[0,1]$ of $\gamma_{0}(z)=1$. Then, $\gamma(1)=z_{0}$. Moreover, using the first in (\ref{zx})
\begin{equation}
\gamma^{\prime}(1)=(\gamma_{0}^{\prime}(z_{0}))^{-1}=\left(\frac{1}{\gamma(1)}+\frac{1}{2\pi i}\int_{0}^{2\pi}\log\{\rho(\psi(\omega))\}\frac{2e^{i\omega}}{(e^{i\omega}-\gamma(1))^{2}}d\omega\right)^{-1},
\label{cv}
\end{equation}
which can be obtained numerically by using the Trapezoidal rule for the integral on the right-hand side of (\ref{cv}).

Clearly, the numerical computation of the exact conformal mappings is generally time consuming. Since $\mathcal{M}$, $\mathcal{L}$ are close to ellipses, alternatively, we can approximate them by conformal mappings that map the interior of ellipses to $G_{\mathcal{C}}$ \cite{neh}. In particular, we can approximate the contour $\mathcal{M}$ by ellipse $\mathcal{E}$ with semi-axes $\rho(0)$, $\rho(\pi/2)$. Then, $\epsilon(x)$ maps $G_{\mathcal{E}}$ to $G_{\mathcal{C}}$ \cite{neh}, where
\begin{displaymath}
\begin{array}{rl}
\epsilon(x)=\sqrt{k}sn\left(\frac{2Q}{\pi}\sin^{-1}(\frac{x}{\sqrt{\rho^{2}(0)-\rho^{2}(\pi/2)}});k^{2}\right),&k=16q\prod_{n=1}^{\infty}\left(\frac{1+q^{2n}}{1+q^{2n-1}}\right)^{8},\\
q=\left(\frac{\rho(0)-\rho(\pi/2)}{\rho(0)+\rho(\pi/2)}\right)^{2},&Q=\int_{0}^{1}\frac{dt}{\sqrt{(1+t^{2})(1-k^{2}t^{2})}},
\end{array}
\end{displaymath} 
where $sn(w;l)$ is the Jacobian elliptic function. Our approximation for $\gamma(x)$ is $\epsilon(x)$, $x\in \mathcal{M}\cup G_{\mathcal{M}}$.
\bibliographystyle{spmpsci}
\bibliography{bibliography}
\end{document}